\documentclass[conference,10pt,compsoc]{IEEEtran}
\ifCLASSOPTIONcompsoc
  \usepackage[nocompress]{cite}
\else
  \usepackage{cite}
\fi
\usepackage{graphicx}
\usepackage{float}
\usepackage{xspace}
\usepackage[dvipsnames]{xcolor}
\usepackage{mathtools}
\usepackage{ltl}
\usepackage{scalefnt}
\usepackage{cleveref}
\usepackage{amsthm}
\usepackage{todo}
\usepackage{caption}
\usepackage{subcaption}
\usepackage{multirow}

\usepackage{tikz}
\usetikzlibrary{calc,3d,fit,patterns}

\newtheorem{thm}{Theorem}[section]
\newtheorem{lem}[thm]{Lemma}

\newtheorem{cor}{Corollary}[section]

\theoremstyle{definition}
\newtheorem{defn}{Definition}[section]

\newcommand{\nat}{\mathbb{N}}

\newcommand{\fo}{FO[$<$]\@\xspace}
\newcommand{\foe}{\mbox{FO[$<,E$]}\@\xspace}
\newcommand{\seins}{S1S[$E$]\@\xspace}
\newcommand{\mple}{MPL[$E$]\@\xspace}
\newcommand{\msoe}{MSO[$E$]\@\xspace}

\newcommand{\U}{\LTLuntil}
\newcommand{\X}{\LTLnext}
\newcommand{\G}{\LTLglobally}
\newcommand{\F}{\LTLeventually}
\newcommand{\R}{\LTLrelease}
\newcommand{\true}{\mathit{true}}

\newcommand{\pathvars}{\mathcal{V}}
\newcommand{\pathassign}{\Pi}
\newcommand{\set}[1]{\{ #1 \}}
\newcommand{\ldot}{\mathpunct{.}}

\newcommand{\K}{\mathcal{K}}

\newcommand{\Veins}{\mathcal{V}_1}
\newcommand{\Vzwei}{\mathcal{V}_2}

\newcommand{\finf}{D_\mathit{inf}}
\newcommand{\finff}{D_\mathit{inf'}}
\newcommand{\qinf}{F_\mathit{inf}}
\newcommand{\qinff}{F_\mathit{inf'}}

\newcommand\ScaleExists[1]{\vcenter{\hbox{\scalefont{#1}$\exists$}}}

\DeclareMathOperator*\bigexists{%
	\vphantom\sum
	\mathchoice{\ScaleExists{2}}{\ScaleExists{1.4}}{\ScaleExists{1}}{\ScaleExists{0.75}}}

\ifCLASSINFOpdf
\else
\fi
\begin{document}
%
\title{The Hierarchy of Hyperlogics\thanks{This work was partially supported by the German Research Foundation (DFG) as part of the Collaborative Research Center ``Methods and Tools for Understanding and Controlling Privacy'' (CRC 1223) and the Collaborative Research Center ``Foundations of Perspicuous Software Systems'' (TRR 248, 389792660), and by the European Research Council (ERC) Grant OSARES (No. 683300).}}

\author{\IEEEauthorblockN{Norine Coenen,
Bernd Finkbeiner,
Christopher Hahn and
Jana Hofmann}
\IEEEauthorblockA{Reactive Systems Group, Saarland University\\
Saarbr\"ucken, Germany\\
Email: \{lastname\}@react.uni-saarland.de}}

 \IEEEoverridecommandlockouts
\IEEEpubid{\makebox[\columnwidth]{978-1-7281-3608-0/19/\$31.00~
		\copyright2019 IEEE \hfill} \hspace{\columnsep}\makebox[\columnwidth]{ }}
\maketitle

\begin{abstract}
	Hyperproperties, which generalize trace properties by relating multiple traces, are widely studied in information-flow security.
	Recently, a number of logics for hyperproperties have been proposed, and there is a need	to understand their decidability and relative expressiveness.
	The new logics have been obtained from standard logics with two principal extensions: temporal logics, like LTL and CTL$^*$, have been generalized to hyperproperties by adding variables for traces or paths. 
	First-order and second-order logics, like monadic first-order logic of order and MSO, have been extended with the equal-level predicate. 
	We study the impact of the two extensions across the spectrum of linear-time and branching-time logics, in particular for logics with quantification over propositions.
	The resulting hierarchy of hyperlogics differs significantly from the classical hierarchy, suggesting that the  equal-level predicate adds more expressiveness than trace and path variables. 
	Within the hierarchy of hyperlogics, we identify new boundaries on the decidability of the satisfiability problem. 
	Specifically, we show that while HyperQPTL and HyperCTL$^*$ are both undecidable in general, formulas within their $\exists^*\forall^*$~fragments are decidable.
\end{abstract}


%
\IEEEpeerreviewmaketitle

\section{Introduction}

Temporal logics are classified into linear-time and branching-time logics:
While \emph{linear-time} temporal logics like LTL~\cite{LTL} describe properties of individual traces, \emph{branching-time} temporal logics like CTL$^*$~\cite{CTLStar} describe properties of computation trees, where the branches can be inspected by quantifying existentially or universally over paths.  
\emph{Hyperlogics} add a second, orthogonal, dimension to this classification~\cite{SpectrumOfTemporalLogics}: while the standard temporal logics only refer to a \emph{single} trace or path at a time, the temporal hyperlogics HyperLTL and HyperCTL$^*$ relate \emph{multiple} traces or paths to each other~\cite{HyperLTL}.  
This makes it possible to express information-flow properties such as noninterference~\cite{Goguen+Meseguer/1982/SecurityPoliciesAndSecurityModels} and observational determinism~\cite{Zdancewic+Myers/03/ObservationalDeterminism}. 
Technically, the temporal hyperlogics extend the standard logics with trace or path variables. 
By quantifying over multiple variables, the formula can refer to several traces or paths at the same time. 
For example, the HyperLTL formula
\begin{equation}
  \forall \pi. \forall \pi'.~ \G \bigwedge_{a \in \mathit{AP}}a_\pi \leftrightarrow a_{\pi'}
   \label{hyperltlexample}
\end{equation}
expresses that \emph{all pairs} of traces must agree on the values of the atomic propositions (given as a set $\mathit{AP}$) at all times.

A different method for the construction of hyperlogics has been introduced for first-order and second-order logics such as monadic first-order logic of order (\fo) and full monadic second-order logic (MSO). 
The extension consists of adding the \emph{equal-level predicate}~$E$ (cf.~\cite{Thomas09,Martin}), which relates the same time points on different traces. 
The HyperLTL formula (\ref{hyperltlexample}), for example, is equivalent to the \foe formula
\[
\forall x. \forall y. ~E(x,y) \rightarrow \bigwedge_{a \in \mathit{AP}}(P_a(x) \leftrightarrow P_a(y)).
\]

It is, so far, poorly understood how these two extensions compare in terms of expressiveness.
A natural point of reference is Kamp's seminal theorem~\cite{Kamp68}, which states (in the formulation of Gabbay et al.~\cite{KampGabbay}) that LTL is expressively equivalent to FO$[<]$.
However, the potential analogue of Kamp's theorem for hyperlogics, that HyperLTL might be equivalent to FO$[<,\,E]$, is known \emph{not} to be true~\cite{Martin}.

In this paper, we initiate a comprehensive study of the spectrum of hyperlogics, guided by the known results for the standard logics (Figure~\ref{fig:lin_trace} and Figure~\ref{fig:bran_trace}).
In addition to the equivalence of LTL and FO$[<]$ established by Kamp's theorem, it is known that quantified propositional temporal logic (QPTL)~\cite{QPTL} and monadic second-order logic of one successor (S$1$S) are expressively equivalent~\cite{QPTL-S1S}.
Moreover, previous work~\cite{moller1999} showed that CTL$^*$ and monadic path logic (MPL)~\cite{abrahamson1980}, as well as quantified computation tree logic (QCTL$^*$)~\cite{QCTLStar} and MSO are expressively equivalent~\cite{laroussinie2014}.

Figures~\ref{fig:lin_hyper}~and~\ref{fig:bran_hyper} show the results for the linear and branching-time hyperlogics, respectively.
For linear time, the most striking difference to the hierarchy of the standard logics is that S1S is no longer equivalent to QPTL when lifted to hyperlogics: \seins is \emph{strictly more expressive} than HyperQPTL, which, in turn, is strictly more expressive than \foe. 
For branching time, we have that HyperQCTL$^*$ is still expressively equivalent to MSO$[E]$. 
However, MPL, which is equivalent to CTL$^*$ in the standard hierarchy (Figure~\ref{fig:bran_trace}), falls strictly \emph{between} HyperCTL$^*$ and HyperQCTL$^*$  
when equipped with the equal-level predicate (\mple).

The choice of logics considered in our expressiveness study is motivated by practical interest in certain hyperproperties. 
Branching-time hyperlogics, for example, are useful to state that a system can generate secret information~\cite{SpectrumOfTemporalLogics}, e.g., there is, at some point, a branching into observably equivalent paths that differ in the values of a secret:
\begin{equation*}
  \exists \pi. \F \exists \pi'.~ (\G \bigwedge_{a \in \mathit{P}}a_\pi \leftrightarrow a_{\pi'}) \land (\X \bigvee_{a \in \mathit{S}} a_\pi \nleftrightarrow a_{\pi'}),
   \label{hyperctlexample}
\end{equation*}
where the set of atomic propositions divides into the two disjoint sets of publicly observable propositions $\mathit{P}$ and secret propositions $\mathit{S}$.

An example for the usefulness of quantification over atomic propositions is the extension of LTL to QPTL, which improves the expressiveness from non-counting properties to general $\omega$-regular properties. In HyperQPTL, the extension
of HyperLTL with quantification over atomic propositions,  it additionally becomes possible to express properties like \emph{promptness}~\cite{Kupferman2009}, which states that there is a bound, common for all traces, on the number of steps until an eventuality $a$ is satisfied. 
Promptness can be expressed in HyperQPTL by using a quantified atomic proposition $q$, such that the first occurrence of $q$ represents the bound:
\begin{equation*}
  \exists q. \forall \pi.~ \F q \land (\neg q \U a_\pi).
   \label{hyperqptlexample}
\end{equation*}

Comparing the impact of adding quantification over trace and path variables to temporal logics vs. adding the equal-level predicate to first-order and second-order logics, our analysis indicates that the equal-level predicate adds more expressive power to a first-order (or second-order) logic than trace and path quantification adds to a temporal logic.
Regarding the decidability of the logics, we show that, in the linear-time hierarchy, the boundary of the decidable hyperproperties can be characterized as the $\exists^* \forall^*$~HyperQPTL fragment. 
(The $\exists^* \forall^*$~fragment consists of all formulas with an arbitrary number of existential trace quantifiers followed by an arbitrary number of universal trace quantifiers.)
In the branching-time hierarchy, we show that $\exists^* \forall^*$~HyperCTL$^*$ formulas can still be decided.
The decidability results are summarized in Table~\ref{tab:satResults}.

The remainder of this paper is structured as follows. 
Section~\ref{sec:preliminaries} covers preliminaries, including the basic temporal logics LTL and CTL$^*$.
In Section~\ref{sec:expressivity}, we provide proofs for the expressiveness results from Figure~\ref{fig:hier}.
In Section~\ref{sec:satisfiability}, we show new decidability bounds in the linear-time and branching-time hyperlogic hierarchies as indicated in Table~\ref{tab:satResults}.

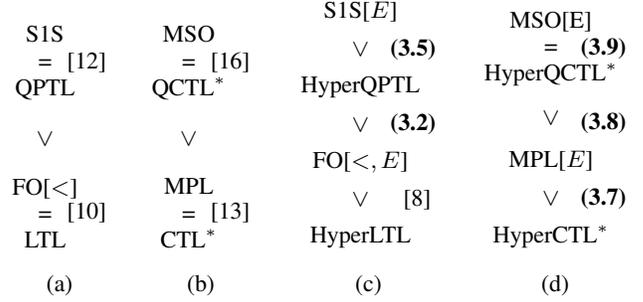
\begin{figure}[t]
\begin{subfigure}[t]{.21\columnwidth}
	\centering
	\begin{tikzpicture}[every node/.append style = {font=\small}]
	\node[align=center] (a) at (0,2) {S1S \\ = \\ QPTL};
	\node[align=center] (c) at (0,0) {\fo \\ = \\ LTL};
	
	\draw[-, color = white] 
	(a) edge node[sloped, color = black] {$>$} (c);
	\node (z) at (0.5,2) {\cite{QPTL-S1S}};
	\node (z) at (0.5,0) {\cite{KampGabbay}};
	\end{tikzpicture}
	\caption{}
	\label{fig:lin_trace}
\end{subfigure}
\begin{subfigure}[t]{.21\columnwidth}
	\centering
	\begin{tikzpicture}[every node/.append style = {font=\small}]
		\node[align=center] (b) at (0,2) {MSO \\ = \\ QCTL$^*$};
		\node[align=center] (a) at (0,0) {MPL \\ = \\ CTL$^*$};
		
		\draw[-, color = white] (a) edge node[sloped, color = black] {$<$} (b);
	\node (z) at (0.5,2) {\cite{laroussinie2014}};
	\node (z) at (0.5,0) {\cite{moller1999}};
	\end{tikzpicture}
	\caption{}
	\label{fig:bran_trace}
\end{subfigure}
\begin{subfigure}[t]{.29\columnwidth}
	\centering
	\begin{tikzpicture}[every node/.append style = {font=\small}]
		\node (f) at (0,3) {\seins};
		\node (g) at (0,2) {HyperQPTL};
		\node (h) at (0,1) {\foe};
		\node (j) at (0,0) {HyperLTL};
		
		\draw[-, color = white] (j) edge node[sloped, color = black] {$<$} (h)
				(h) edge node[sloped, color = black] {$<$} (g)
				 (g) edge node[sloped, color = black] {$<$} (f);
	\node (z) at (0.7,2.5) {\bf(\ref{S1SHyperQPTL})};
	\node (z) at (0.7,1.5) {\bf(\ref{HyperQPTLmoreExpressiveFOE})};
	\node (z) at (0.7,0.5) {\cite{Martin}};
	\end{tikzpicture}
	\caption{}
	\label{fig:lin_hyper}
\end{subfigure}
\begin{subfigure}[t]{.29\columnwidth}
	\centering
	\begin{tikzpicture}[every node/.append style = {font=\small}]
		\node[color=white] (h) at (0,3) {MSO[E]};
		\node[align=center] (h) at (0,2.5) {MSO[E] \\ = \\ HyperQCTL$^*$};
		\node (i) at (0,1) {\mple};
		\node (j) at (0,0) {HyperCTL$^*$};
		
		\draw[-, color = white] (j) edge node[sloped, color = black] {$<$} (i)
				 (i) edge node[sloped, color = black] {$<$} (h);
	\node (z) at (0.7,2.5) {\bf (\ref{thm:HyperQCTL*equivMSOE})};
	\node (z) at (0.7,1.5) {\bf(\ref{HyperQCTLStarMPLE})};
	\node (z) at (0.7,0.5) {\bf(\ref{MPLEHyperCTLStar})};
	\end{tikzpicture}
	\caption{}
	\label{fig:bran_hyper}
\end{subfigure}
\caption{
	The linear-time hierarchies of standard logics (a) and hyperlogics (c), and the branching-time hierarchies of standard logics (b) and hyperlogics (d). Novel results are annotated with the corresponding theorem number.
}
\label{fig:hier}
\end{figure}

\begin{table}
\small
\setlength{\tabcolsep}{6pt}
\def\arraystretch{1.08}
\centering
\scalebox{0.96}{
\begin{tabular}[t]{ll||ll}
	Logic & Result & Hyperlogic & Result \\ \hline
	\multirow{2}{*}{LTL}      & \multirow{2}{*}{decidable~\cite{LTL-SAT}} & \multirow{2}{*}{HyperLTL} & undecidable~\cite{DecidingHyperLTL} \\
	& & & $\exists^* \forall^*$ decidable~\cite{DecidingHyperLTL} \\[\smallskipamount]
	\multirow{2}{*}{QPTL}      & \multirow{2}{*}{decidable~\cite{QPTL}} & \multirow{2}{*}{HyperQPTL} & undecidable {\bf(\ref{thm:HyperQPTLUndec})} \\ 
	& & & $\exists^* \forall^*$ decidable {\bf(\ref{hyperqptl-sat-decidable})}\\[\smallskipamount]
	\multirow{2}{*}{CTL$^*$}      & \multirow{2}{*}{decidable~\cite{CTLStar-SAT}} & \multirow{2}{*}{HyperCTL$^*$} & undecidable {\bf(\ref{HyperCTLStarUndec})} \\ 
	& & & $\exists^* \forall^*$ decidable {\bf(\ref{HyperCTLStarSAT})} 
\end{tabular}}
\caption{Satisfiability results. Novel results are annotated with the corresponding theorem number.}
\label{tab:satResults}
\end{table}

\section{Preliminaries}
\label{sec:preliminaries}

We formally define traces and trees, and introduce some basic notation for trace and path manipulation.
We define the temporal logics LTL and CTL$^*$, which serve as the foundation for the extensions studied in this paper.

\subsection{Traces and Trees}
Let $\mathit{AP}$ be a set of atomic propositions. 
We call an infinite sequence over subsets of atomic propositions $t \in (2^\mathit{AP})^\omega$ a \emph{trace}. 
A set of traces $T \subseteq (2^\mathit{AP})^\omega$ is also called a \emph{trace property} while a set of sets of traces $H \subseteq 2^{((2^\mathit{AP})^\omega)}$ is called a \emph{hyperproperty}. 
Trace manipulation is defined as follows: for a trace $t$ and a natural number $i \geq 0$, we denote the $i$-th element of the trace by $t[i]$.
For a natural number $j \geq i$, $t[i,j]$ denotes the sequence $t[i]t[i+1] \ldots t[j-1]t[j]$. 
Moreover, $t[i,\infty]$ denotes the infinite suffix of $t$ starting at position $i$.
For two traces $t$ and $t'$, we define a zipping operation as follows: $zip(t,t') = (t[0],t'[0]) (t[1],t'[1]) \ldots$.

A \emph{tree} $\mathcal{T}$ is defined as a partially-ordered infinite set of nodes $S$, where all nodes share a common minimal element $r \in S$, called the \emph{root} of the tree.
Moreover, for every node $s \in S$, the set of its \emph{ancestors} $\set{s' | s' < s}$ is totally-ordered.
We say that $s'$ is the \emph{direct ancestor} of $s$, if $s' < s$, and there is no $s''$ such that $s' < s'' < s$.
A \emph{$\Sigma$-labeled tree} is defined as a tree $\mathcal{T}$ equipped with a function $L: S \to \Sigma$, that labels every node with an element from a finite set $\Sigma$.
For the case that $\Sigma = 2^\mathit{AP}$, we say that the tree is \emph{$\mathit{AP}$-labeled}.
A \emph{path} through a tree $\mathcal{T}$ is a sequence $\sigma = s_0, s_1, \ldots$ of direct ancestors in $\mathcal{T}$, i.e., for all $s_i, s_{i+1}$, node $s_{i}$ is the direct ancestor of $s_{i+1}$.
A path is called \emph{initial} if $s_0$ is the root node, which we omit if it is clear from the context.
We use the same path manipulation operations as for traces.
The set of \emph{paths originating in node $s$} $\in S$ is denoted by $Paths(\mathcal{T},s)$.
If $s$ is the root node, we simply write $Paths(\mathcal{T})$.

\subsection{LTL and CTL$^*$}
Linear-time Temporal Logic (LTL)~\cite{LTL} and Computation Tree Logic (CTL$^*$)~\cite{CTLStar} are the most studied temporal logics for linear-time and branching-time properties, respectively.
\begin{defn}[LTL]
	\label{ltl}
LTL is a linear-time logic that combines the usual Boolean connectives with temporal modalities $\LTLnext$~(next) and $\LTLuntil$~(until).
The syntax is given by the following grammar:
\begin{align*}
\varphi &\Coloneqq a ~|~ \neg \varphi ~|~ \varphi \lor \varphi ~|~ \X \psi ~|~ \psi \U \psi,
\end{align*}
where $a \in AP$, which is the set of atomic propositions. We allow the standard Boolean connectives $\wedge$, $\rightarrow$, $\leftrightarrow$ as well as the derived LTL modalities release $\varphi \R \psi \equiv \neg(\neg\varphi \U \neg\psi)$, eventually $\F \varphi \equiv \true \U \varphi$, and globally $\G \varphi \equiv \neg \F \neg \varphi$.
Given a trace $t \in (2^\mathit{AP})^\omega$, the semantics of an LTL formula is defined as follows:
\begin{alignat*}{3}
t &\models a                             &&\text{ iff } && a \in t[0] \\
t &\models \neg \varphi                  &&\text{ iff } && t \not\models \varphi \\
t &\models \varphi_1 \vee \varphi_2      &&\text{ iff }&&t \models \varphi_1 \text{ or } t \models \varphi_2 \\
t &\models \LTLnext \varphi              &&\text{ iff } && t[1,\infty] \models \varphi \\
t &\models \varphi_1 \LTLuntil \varphi_2 &&\text{ iff } && \exists i \geq 0.~ t[i,\infty] \models \varphi_2 \\
&                                        &&             && \land \forall 0 \leq j < i.~ t[j,\infty] \models \varphi_1.
\end{alignat*}
\end{defn}

\begin{defn}[CTL$^*$]
	\label{ctl}
CTL$^*$ is a branching-time logic (i.e. the model is a tree) that extends LTL (Def.~\ref{ltl}) with a path quantifier E meaning ``there exists a path''.
The syntax, where $\varphi$ denotes state formulas and $\psi$ denotes path formulas, is given as follows:
\begin{align*}
\varphi &\Coloneqq a ~|~ \neg \varphi ~|~ \varphi \lor \varphi ~|~ \text{E}\psi\\
\psi &\Coloneqq \varphi ~|~ \neg \psi ~|~ \psi \lor \psi ~|~ \X \psi ~|~ \psi \U \psi,
\end{align*}
where $a \in AP$ is an atomic proposition.
The semantics of CTL$^*$ is defined over an $\mathit{AP}$-labeled tree $\mathcal{T}$ with nodes $S$ and labeling function $L$.
Given a node $s \in S$ and a path $p$ in $\mathcal{T}$, we define the semantics of CTL$^*$ state and path formulas as follows:
\begin{alignat*}{3}
	s &\models_\mathcal{T} a &&\text{ iff } && a \in L(s) \\
	s &\models_\mathcal{T} \neg \varphi         &&\text{ iff } && s \not\models_\mathcal{T} \varphi \\
	s &\models_\mathcal{T} \varphi_1 \vee \varphi_2 &&\text{ iff }&&s \models_\mathcal{T} \varphi_1 \text{ or } s \models_\mathcal{T} \varphi_2 \\
	s &\models_\mathcal{T} \text{E} \psi   &&\text{ iff } && \exists p \in \mathit{Paths}(\mathcal{T}, s).~ p \models_\mathcal{T} \psi  \\
	p &\models_\mathcal{T} \varphi &&\text{ iff } &&p[0] \models_\mathcal{T} \varphi \\
	p &\models_\mathcal{T} \neg \psi         &&\text{ iff } &&p \not\models_\mathcal{T} \psi \\
	p &\models_\mathcal{T} \psi_1 \vee \psi_2 &&\text{ iff }&&p \models_\mathcal{T} \psi_1 \text{ or } p \models_\mathcal{T} \psi_2 \\
	p &\models_\mathcal{T} \X \psi &&\text{ iff } &&p[1,\infty] \models_\mathcal{T} \psi \\
	p &\models_\mathcal{T} \psi_1 \U \psi_2 &&\text{ iff }&& \exists i \geq 0.~ p[i,\infty] \models_\mathcal{T} \psi_2 \\
	& && && \land \forall 0 \leq j < i.~ p[j,\infty] \models_\mathcal{T} \psi_1.
\end{alignat*}
For a tree $\mathcal{T}$ and a CTL$^*$ formula $\varphi$, we write $\mathcal{T} \models \varphi$ if $\mathcal{T}$ has root $r$, such that $r \models_{\mathcal T} \varphi$.
\end{defn}

\section{Expressivity}
\label{sec:expressivity}

We examine the expressive power of the hyperlogics by going bottom-up through the linear-time and branching-time hierarchies of hyperlogics depicted in Figure~\ref{fig:hier}.

\subsection{The Linear-Time Hierarchy}
In the following, we first show that, like for the standard logics, HyperQPTL is strictly more expressive than \foe. We then establish that, indeed, \seins is strictly more expressive than HyperQPTL.

\begin{defn}[HyperLTL]
	\label{hyperltl}
HyperLTL~\cite{HyperLTL} extends LTL (Def.~\ref{ltl}) with explicit trace quantification.
Let $\pathvars = \set{\pi_1, \pi_2, \ldots}$ be an infinite set of trace variables.
HyperLTL formulas are defined by the grammar:
\begin{align*}
\varphi &{}\Coloneqq \forall\pi\ldot\varphi \mid \exists\pi\ldot\varphi \mid \psi \enspace \\
\psi &{}\Coloneqq a_\pi \mid \neg\psi \mid \psi\lor\psi \mid \X\psi \mid \psi\U\psi,
\end{align*}
where $a \in AP$ and $\pi \in \pathvars$.
Here, $\forall\pi\ldot\varphi$ and $\exists\pi\ldot\varphi$ denote universal and existential trace quantification, and $a_\pi$ requires the atomic proposition $a$ to hold on trace $\pi$. 
The semantics of HyperLTL is defined with respect to a set of traces $T$.
Let $\pathassign : \pathvars \to T$ be a trace assignment that maps trace variables to traces in $T$. 
We can update a trace assignment $\Pi$, denoted by $\pathassign[\pi \mapsto t]$, where $\pi$ maps to $t$ and all other trace variables are as in $\pathassign$.
The satisfaction relation $\models_T$ for HyperLTL over a set of traces $T$ is defined as follows:
\begin{alignat*}{3}
&\pathassign,i \models_T a_\pi       &&\text{ iff } &&a \in \pathassign(\pi)[i] \\
&\pathassign,i \models_T \neg \varphi              && \text{ iff } &&\pathassign,i \not\models_T \varphi \\
&\pathassign,i \models_T \varphi_1 \lor \varphi_2         && \text{ iff } &&\pathassign,i \models_T \varphi_1 \text{ or } \pathassign,i \models_T \varphi_2 \\
&\pathassign,i \models_T \X \varphi                && \text{ iff } &&\pathassign,i+1 \models_T \varphi \\
&\pathassign,i \models_T \varphi_1 \U \varphi_2             && \text{ iff } &&\exists j \geq i \ldot ~ \pathassign, j \models_T \varphi_2 \\
& && && \phantom{\exists j \geq i} \land \forall i \leq k < j \ldot ~\pathassign,k \models_T \varphi_1 \\
&\pathassign,i \models_T \exists \pi \ldot \varphi && \text{ iff } && \exists t \in T.~ \pathassign[\pi \mapsto t], i \models_T \varphi\\
&\pathassign,i \models_T \forall \pi \ldot \varphi && \text{ iff } && \forall t \in T.~ \pathassign[\pi \mapsto t], i \models_T \varphi.
\end{alignat*}
\end{defn}
We say that a trace set $T$ satisfies a HyperLTL formula $\varphi$, written as $T \models \varphi$, if $\emptyset,0 \models_T \varphi$, where $\emptyset$ denotes the empty trace assignment.

\begin{defn}[HyperQPTL]
	\label{hyperqptl}
HyperQPTL~\cite{MarkusThesis} extends HyperLTL (Def.~\ref{hyperltl}) with explicit quantification over atomic propositions.
We add atomic formulas $q$, which are independent of the trace variables, and prenex propositional quantification $\exists q. \varphi$ to the syntax.
HyperQPTL inherits the semantics of HyperLTL with two additional rules for the new syntactic constructs:
\begin{alignat*}{3}
&\pathassign,i \models_T \exists q \ldot \varphi &&\text{ iff } && \exists t \in (2^{\set{q}})^\omega.~ \pathassign[\pi_q \mapsto t],i \models_T \varphi \\
&\pathassign,i \models_T q  &&\text{ iff } &&q \in \pathassign(\pi_q)[i]. 
\end{alignat*}
\end{defn}

\begin{defn}[\foe]
	\label{foe}
\foe extends \fo with the equal-level predicate $E$, which relates points in time.
The syntax of \foe is obtained by extending the syntax of \fo with $E(x,y)$. Given a set of atomic propositions $AP$ and a set $V_1$ of first-order variables, we define the syntax of \foe formulas as follows:
\begin{align*}	
\tau  &\Coloneqq P_a(x) ~|~ x<y ~|~ x=y ~|~ E(x,y) \\
\varphi &\Coloneqq \tau ~|~ \neg \varphi ~|~ \varphi_1 \vee \varphi_2  ~|~ \exists x. \varphi,
\end{align*}
where $a \in AP$ and $x, y \in V_1$.
\end{defn}
While \fo formulas are interpreted over a trace $t$, we interpret an \foe formula $\varphi$ over a set of traces $T$, writing $T \models \varphi$ if $T$ satisfies $\varphi$. As first described in~\cite{Martin}, we assign first-order variables with elements from the domain $T \times \nat$. The $<$ relation is defined as the set $\{(t,n_1), (t,n_2) \in (T \times \nat)^2 ~|~ n_1 < n_2 \}$ and the equal-level predicate is defined as $\{ (t_1, n), (t_2, n) \in (T \times \nat)^2\}$. Note that $x < y$ holds in \foe iff $y$ is a successor of $x$ on the same trace.

\begin{lem}
	\label{HyperQPTLsubsumesFOE}
	HyperQPTL is at least as expressive as \foe.
\end{lem}
\begin{proof}
	We give a linear translation from \foe to HyperQPTL. We encode each first-order variable $x$ as a combination of a trace variable $\pi_x$ and a propositional variable $q_x$, where we enforce $q_x$ to hold exactly once.
	Let $\varphi$ be an \foe formula over $AP$ in prenex normal form.
	We construct the HyperQPTL formula hq($\varphi$) as follows:
	\begin{alignat*}{3}
	&\text{hq}(P_a(x)) &&=~&& \LTLfinally (q_x \land a_{\pi_x}) \\
	&\text{hq}(x < y) &&=~&& \LTLfinally (q_x \land \LTLnext \LTLfinally q_y) \land \LTLglobally (\bigwedge_{a \in AP} a_{\pi_x} \leftrightarrow a_{\pi_y}) \\
	&\text{hq}(x = y) &&=~&& \LTLfinally (q_x \land  q_y) \land \LTLglobally (\bigwedge_{a \in AP} a_{\pi_x} \leftrightarrow a_{\pi_y} )\\
	&\text{hq}(E(x,y)) &&=~&& \LTLfinally (q_x \land  q_y) \\
	&\text{hq}(\neg \varphi_1) &&=~&& \neg \text{hq}(\varphi_1) \\
	&\text{hq}(\varphi_1 \lor \varphi_2) &&=~&& ~\text{hq}(\varphi_1) \lor \text{hq}(\varphi_2) \\
	&\text{hq}(\exists x. \varphi_1) &&=~&& ~\exists \pi_x.~ \exists q_x.~ (\neg q_x) \LTLuntil (q_x \land \LTLnext \LTLglobally \neg q_x ) \\
	& && && \quad \land \text{hq}(\varphi_1).
	\end{alignat*}
	Note that since $T$ is a set of traces, two traces are equal in HyperQPTL iff they globally agree on all atomic propositions.
	Since we assume $\varphi$ to be in prenex normal form, hq($\varphi$) is a valid HyperQPTL formula.
	A straight-forward induction proves that for all trace sets $T$, $T \models \varphi$ iff $T \models \text{hq}(\varphi)$.
\end{proof}

\begin{thm}
\label{HyperQPTLmoreExpressiveFOE}
	HyperQPTL is strictly more expressive than \foe.
\end{thm}
\begin{proof}
	With Lemma~\ref{HyperQPTLsubsumesFOE}, we are left to show that there are properties that HyperQPTL can express, but \foe cannot.
	We apply a similar technique as in~\cite{MarkusThesis}: 
	Consider the class of models $T$ with only a single trace.
	In this class, HyperQPTL is expressively equivalent to QPTL and \foe is expressively equivalent to \fo, which is equivalent to LTL~\cite{Kamp68}.
	It is known, however, that QPTL is strictly more expressive than LTL since QPTL can express any $\omega$-regular language~\cite{sistla1987complementation}, which LTL cannot~\cite{mcnaughton1971}. Therefore, HyperQPTL must be strictly more expressive than \foe.
\end{proof}

\begin{defn}[{S1S[E]}]
Similar to the definition of \foe, we extend S1S with the equal-level predicate and interpret it over sets of traces. 
Let $AP$ be a set of atomic propositions, $V_1 = \{x_1, x_2, \ldots\}$ be a set of first-order variables, and $V_2 = \{X_1, X_2, \ldots\}$ a set of second-order variables.
The syntax of \seins formulas $\varphi$ is defined as follows:
\begin{align*}
\tau &\Coloneqq x \mid \mathit{min}(x) \mid S(\tau)\\
\varphi &\Coloneqq \tau \in X \mid \tau = \tau \mid E(\tau,\tau) \mid \neg \varphi \mid \varphi \vee \varphi \mid \exists x. \varphi \mid \exists X. \varphi,
\end{align*}
where $x \in V_1$ is a first-order variable, $S$ denotes the successor relation, and $\mathit{min}(x)$ indicates the minimal element of the traces addressed by $x$.
Furthermore, $E(\tau, \tau)$ is the equal-level predicate and $X \in V_2 \cup \{X_a ~|~ a \in AP \}$.
We interpret \seins formulas over a set of traces $T$. 
As in the case of \foe, the domain of the first-order variables is $T \times \nat$.
Let $\mathcal{V}_1: V_1 \to T \times \mathbb{N}$ and $\mathcal{V}_2: V_2 \to 2^{(T\times \mathbb{N})}$ be the first-order and second-order valuation, respectively. 
The value of a term is defined as:
\begin{align*}
[x]_{\Veins} &= \Veins(x)\\
[\mathit{min}(x)]_{\Veins} &= (\mathit{proj}_1 (\mathcal{V}_1(x)),0)\\
[S(\tau)]_{\Veins} &= (\mathit{proj}_1([\tau]_{\Veins}) , \mathit{proj}_2([\tau]_{\Veins}) +1),
\end{align*}
where $\mathit{proj}_1$ and $\mathit{proj}_2$ denote the projection to the first and second component, respectively.
Let $\varphi$ be an \seins formula with free first-order and second-order variables $V'_1 \subseteq V_1$ and $V'_2 \subseteq V_2 \cup \{X_a ~|~ a \in AP \}$, respectively.
We define the satisfaction relation $\Veins, \Vzwei \models \varphi $ with respect to two valuations $\mathcal V_1, \mathcal V_2$ assigning all free variables in $V'_1$ and $V'_2$ as follows:
\begin{alignat*}{3}
\Veins, \Vzwei &\models_T \tau \in X              &&\text{ iff } &&[\tau]_{\Veins} \in \Vzwei(X) \\
\Veins, \Vzwei &\models_T \tau_1 =\tau_2                &&\text{ iff } &&[\tau_1]_{\Veins} = [\tau_2]_{\Veins} \\
\Veins, \Vzwei &\models_T E(\tau_1,\tau_2)               &&\text{ iff } &&\mathit{proj}_2([\tau_1]_{\Veins}) = \mathit{proj}_2([\tau_2]_{\Veins}) \\
\Veins, \Vzwei &\models_T \neg \varphi         &&\text{ iff } &&\Veins, \Vzwei \not\models_T \varphi \\
\Veins, \Vzwei &\models_T \varphi_1 \vee \varphi_2 &&\text{ iff } &&\Veins, \Vzwei \models_T \varphi_1 \text{ or } \Veins, \Vzwei \models_T \varphi_2 \\
\Veins, \Vzwei &\models_T \exists x. \varphi   &&\text{ iff } && \exists (t, n) \in T \times \nat.~ \\
& && && \Veins[x \mapsto (t, n)], \Vzwei \models_T \varphi \\
\Veins, \Vzwei &\models_T \exists X. \varphi   &&\text{ iff } && \exists A \subseteq T \times \nat.~ \\
& && && \Veins, \Vzwei[X \mapsto A] \models_T \varphi,
\end{alignat*}
where $\mathcal{V}_i[x \mapsto v]$ updates a valuation. 
\end{defn}
We call an \seins formula $\varphi$ closed if every free variable is a second-order variable of the form $X_a$ with $a \in AP$.
We say that a trace set $T$ over $AP$ satisfies a closed \seins formula $\varphi$, written $T \models \varphi$, if $\emptyset, \mathcal{V}_2 \models_T \varphi$, where $\emptyset$ denotes the empty first-order valuation and $\mathcal{V}_2$ assigns each free $X_a$ in $\varphi$ to the set $\{(t,n) \in T \times \nat~|~ a \in t[n]\}$.

\begin{lem}
	\label{S1SESubsumesHyperQPTL}
	\seins is at least as expressive as HyperQPTL.
\end{lem}

\begin{proof}
	We describe a linear translation from HyperQPTL to \seins, which is similar to the translation from HyperLTL to \foe described in~\cite{Martin}.
	Assume two sets of first-order variables $Tr = \set{x_\pi, x_{\pi'}, \ldots}$ for trace variables and $Ti = \set{y_1, y_2, \ldots}$ to indicate time.
	Additionally, we represent propositional variables with second-order variables.
	Given a HyperQPTL formula $\varphi$ over atomic propositions $AP$ and a time variable $y_i$, we inductively construct the \seins formula with free second-order variables $\{X_a ~|~ a \in AP \}$ as follows:
	\begin{alignat*}{3}
	&\text{se}(a_{\pi}, y_i) &&=~&& \exists x.~ x \geq x_\pi \lor x < x_\pi \wedge E(y_i, x) \\
	& && && \quad \wedge x \in X_a\\
	\noalign{ \text{where $x_\pi$ is the trace variable for $\pi$ }}
	&\text{se}(q, y_i) &&=~&& \exists x_p.~ E(y_i, x_q) \wedge x_q \in X_q\\
	\noalign{ \text{where $X_q$ is the propositional variable for $q$} }
	&\text{se}(\neg \varphi_1, y_i) &&=&& \neg \text{se}(\varphi_1, y_i) \\
	&\text{se}(\varphi_1 \lor \varphi_2, y_i) &&=&& \text{se}(\varphi_1, y_i) \lor \text{se}(\varphi_2, y_i) \\
	&\text{se}(\LTLnext \varphi_1, y_i) &&=&& \text{se}(\varphi_1, S(y_j)) \\
	&\text{se}(\varphi_1 \LTLuntil \varphi_2, y_i) &&=&& \exists y_j \geq y_i.~ \text{se}(\varphi_2, y_j) \\
	& && && \quad \wedge (\forall y_k.~ y_i \leq y_k < y_j \rightarrow \text{se}(\varphi_1, y_k)) \\
	&\text{se}(\exists \pi.\varphi_1, y_i) &&=&& \exists x_\pi.~ \text{se}(\varphi_1,y_i)\\
	\noalign{ \text{where $x_\pi$ is now used as trace variable for $\pi$ }}
	&\text{se}(\exists q.\varphi_1, y_i) &&=&& \exists X_q.~ \text{se}(\varphi_1, y_i) \\
	\noalign{ \text{where $X_q$ is now used as propositional variable for $q$.} }
	\end{alignat*}
	For a HyperQPTL formula $\varphi$, we define
	\begin{equation*}
		\text{se}(\varphi) \coloneqq \exists y_0 \ldot (\neg \exists y \ldot y < y_0) \land \text{se}(\varphi, y_0).
	\end{equation*}	
	Using induction, it follows that for each $\varphi$ and trace set $T$, $T \models \varphi$ iff $T \models \text{se}(\varphi)$.
\end{proof}

\begin{lem}
	\label{S1SEModelChecking}
	The \seins model checking problem is undecidable.
\end{lem}

\begin{proof}
	The \seins model checking problem is to decide for a formula $\varphi$ and a regular trace set $T$, whether $T \models \varphi$. A trace set is regular, if it can be described by a Kripke structure.
	We prove this Lemma by a reduction from $2$-counter machines ($2$CM), which are known to be Turing complete. 
	We describe a simple trace set $T$ such that given a $2$CM $\mathcal M$ with an initial configuration $s_0$, we can construct an \seins formula $\varphi_{\mathcal M, s_0}$ such that $\mathcal M$ halts iff $T \models \varphi_{\mathcal M, s_0}$.
	A $2$CM consists of a finite set of instructions $l_1:\mathit{instr}_1 ; \ldots ; l_{k-1} : \mathit{instr}_{k-1} ; l_{k} : \mathit{instr}_{\mathit halt}$, where the last instruction is the instruction to halt and all other instructions are of one of the following forms:
	\begin{itemize}
		\item $c_i \coloneqq c_i +1 ~;~ \mathtt{goto} ~ l_j$ \hspace{1em} (for $i \in \{1,2\}$ and $1 \leq j \leq k$)
		\item $\mathtt{if} ~ c_i = 0 ~\mathtt{then~goto}~l_j~\mathtt{else}~c_i \coloneqq c_i-1 ; ~ \mathtt{goto} ~ l_{j'}$ \hspace{1em} (for $i \in \{1,2\}$ and $1 \leq j, j' \leq k$).
	\end{itemize}
	A $2$CM configuration $s$ is a triple $(i, m, n)$, which indicates that the values of the two counters are currently $m$ and $n$ and that the next instruction to be executed is $l_i$.
	We call each configuration in which $i$ denotes the halting instruction a halting configuration $s_\mathit{halt}$.
	Furthermore, we say that a $2$CM $\mathcal{M}$ halts for a given initial configuration $s_0$ if there is a finite sequence $s_0, s_1, \ldots , s_{\mathit halt}$ such that for all two successive configurations $s_i, s_{i+1}$, the latter one is a result of applying the instruction specified in $s_i$ to configuration $s_i$.
	The main ideas to encode the halting problem of a $2$CM $\mathcal{M}$ with initial configuration $s_0$ into \seins are the following:
	
	\begin{itemize}
		\item We can express in \seins that a set $X_t$ contains exactly all nodes of a trace $t \in T$.
		\item Each $2$CM configuration $s$ is encoded as a trace $t_s$ over atomic propositions $c_1, c_2$ and $l$ which are true exactly once on the trace. A state where the first counter has value $m$ is encoded as a trace $t$ with $c_1 \in t[m]$.
		\item We choose the trace set $T$ to be the infinite trace set containing a trace $t_s$ for all possible $2$CM configurations $s \in \set{(i,m,n) \in \nat^3}$. Note that $T$ is clearly regular.
		\item We can give an \seins formula $\mathit{succ}(X_{t_s}, X_{t_{s'}})$, which is true for traces $t_s, t_{s'}$ iff configuration $s'$ is the result of applying the instruction $l_i$ to configuration $s = (i, m, n)$.
		\item Given a machine $\mathcal{M}$ with initial configuration $s_0$, we give an \seins formula $\mathit{halting}(X)$ which is true iff $X$ encodes a halting computation of $\mathcal{M}$. The formula $\mathit{halting}(X)$ is a conjunct of the following requirements:
		\begin{itemize}
			\item $X$ is a union of finitely many encodings $X_{t_s}$. This can be formulated in \seins by expressing that there is an upper bound on the positions where $c_1, c_2$, and $l$ occur on traces in $X$.
			\item $X$ is predecessor closed with respect to the instructions of the machine, i.e., if $X_{t_s}$ is a subset of $X$, then either $X_{t_s}$ is the trace encoding of the initial configuration $s_0$, or there is a $X_{t_{s'}} \subseteq X$ such that $\mathit{succ}(X_{t_{s'}}, X_{t_{s}})$.
			\item There is a halting configuration in $X$, i.e., there is a trace $X_{t_{s}} \subseteq X$ where $l$ holds at position $k$ in $t_{s}$.
		\end{itemize}
	\end{itemize}
	Using the ideas presented above, we can define $\varphi_{\mathcal M, s_0}$ as a formula that checks whether there is a subset $X$ of $T$ which encodes a halting computation of $\mathcal M$ starting in $s_0$, i.e., $\varphi_{\mathcal M, s_0} \coloneqq \exists X.~ \mathit{halting}(X)$.
\end{proof}

\begin{thm}\label{S1SHyperQPTL}
	\seins is strictly more expressive than HyperQPTL.
\end{thm}
\begin{proof}
	This follows from Lemma~\ref{S1SESubsumesHyperQPTL} and Lemma~\ref{S1SEModelChecking}, since the HyperQPTL model checking problem is decidable~\cite{MarkusThesis}.
\end{proof}

\subsection{The Branching-Time Hierarchy}
The question studied in this section is whether the equi-expressiveness of CTL$^*$ and MPL, and of QCTL$^*$ and MSO, translates to the corresponding hyperlogics. 
We establish that \mple (even if restricted to bisimulation-invariant properties) is strictly more expressive than HyperCTL$^*$. 
We then show that HyperQCTL$^*$ is more expressive than \mple. 
Lastly, we show the rather surprising result that \msoe is not more expressive than HyperQCTL$^*$. In fact, the two logics are equally expressive. 

\begin{defn}[HyperCTL$^*$]
	\label{hyperctl}
HyperCTL$^*$~\cite{HyperLTL} generalizes CTL$^*$ (Def.~\ref{ctl}) by adding explicit path variables and quantification.
Quantification in HyperCTL$^*$ ranges over the paths in a tree.
Let $\pi \in \mathcal{V}$ be a \emph{path variable} from an infinite supply of path variables $\mathcal{V}$ and let $\exists \pi.~\varphi$ be the explicit existential \emph{path quantification}. 
HyperCTL$^*$ formulas are generated by the following grammar:
\[
\varphi \Coloneqq a_\pi \mid \neg \varphi \mid \varphi \vee \varphi \mid \LTLnext \varphi \mid \varphi \LTLuntil \varphi \mid \exists \pi.~ \varphi.
\]
The semantics of a HyperCTL$^*$ formula are defined with respect to a tree $\mathcal{T}$ and a \emph{path assignment} $\Pi:\mathcal{V} \rightarrow \mathit{Paths}(\mathcal{T})$, which is a partial mapping from path variables to actual paths in the tree. 
The satisfaction relation $\models_{\mathcal T}$ is given as follows:
\begin{alignat*}{3}
&\Pi, i \models_\mathcal{T} a_\pi &&\text{ iff } &&a \in L(\Pi(\pi)[i]) \\
&\Pi, i \models_\mathcal{T} \neg \varphi &&\text{ iff } &&\Pi, i \not \models_\mathcal{T} \varphi \\
&\Pi, i \models_\mathcal{T} \varphi_1 \vee \varphi_2 &&\text{ iff } &&\Pi, i \models_\mathcal{T} \varphi_1 \text{ or } \Pi, i \models_\mathcal{T} \varphi_2 \\
&\Pi, i \models_\mathcal{T} \X \varphi &&\text{ iff } &&\Pi, i + 1 \models_\mathcal{T} \varphi\\
&\Pi, i \models_\mathcal{T} \varphi_1 \U \varphi_2 &&\text{ iff } && \exists j \geq i.~ \Pi, j \models_\mathcal{T} \varphi_2\\
& && && \land \forall i \leq k < j. ~\Pi, k \models_\mathcal{T} \varphi_1\\
&\Pi, i \models_\mathcal{T} \exists \pi. \varphi &&\text{ iff } && \exists p \in \mathit{Paths}(\mathcal{T}).~ p[0, i] = \varepsilon[0, i] \\
& && && \land \Pi[\pi \mapsto p, \varepsilon \mapsto p], i \models_\mathcal{T} \varphi,
\end{alignat*}
where we use $\varepsilon$ to denote the last path that was added to the path assignment $\Pi$. We say that a tree $\mathcal{T}$ satisfies a HyperCTL$^*$ formula $\varphi$, written as $\mathcal{T} \models \varphi$, if $\emptyset,0 \models_\mathcal{T} \varphi$, where $\emptyset$ denotes the empty path assignment.
\end{defn}

Without loss of generality, we assume that HyperCTL$^*$ formulas are given in negation normal form (NNF) where negations only occur directly in front of atomic propositions. 
This can be achieved by a straight-forward extension of the syntax and semantics with \emph{conjunction}, the \emph{universal path quantifier} $\forall \pi.~ \varphi$ and the temporal modality $\LTLrelease$ (release), which has the following semantics~\cite{PrinciplesOfModelChecking}:
\begin{alignat*}{3}
&\Pi, i \models_\mathcal{T} \varphi \R \psi && \text{ iff } && \forall j \geq i.~ \Pi, j \models_\mathcal{T} \psi ~ \lor \\
& && && (\exists j \geq i.~ \Pi, j \models_\mathcal{T} \varphi \\
& && && \land \forall i \leq k \leq j.~ \Pi, k \models_\mathcal{T} \psi).
\end{alignat*}

\begin{defn}[\mple]
MPL equipped with the equal-level predicate (\mple) is, syntactically, \foe (Def.~\ref{foe}) with additional second-order set variables $\set{X, Y, \ldots}$, i.e., with atomic formulas $x \in X$ and second-order quantification $\exists X. \varphi$. 
\mple, interpreted over trees $\mathcal{T}$, maps first-order variables to nodes in the tree and second-order variables to sets of nodes. 
$x < y$ indicates that $x$ is an ancestor of $y$. 
Atomic formulas $x \in X$ and $x = y$ are interpreted as expected as set membership and equality on nodes. 
The equal-level predicate $E(x,y)$ denotes that two nodes $x$ and $y$ are on the same level, i.e., have the same number of ancestors.
As for MPL, we require that \mple's second-order quantification is restricted to full paths, i.e., each quantified second-order variable $X$ is mapped to a set whose nodes constitute exactly one path of the tree.
We write $\mathcal{T} \models \varphi$ if the tree $\mathcal{T}$ satisfies the \mple formula~$\varphi$.
\end{defn}

\begin{lem}
	\label{MPLEsubsumesHyperKCTL*}
	\mple is at least as expressive as HyperCTL$^*$.
\end{lem}
\begin{proof}
	We can give a linear translation from  HyperCTL$^*$ to \mple, which is very similar to the one in Lemma~\ref{S1SESubsumesHyperQPTL}.
	As before, we assume sets of first-order variables \mbox{$Tr = \set{x_1, x_2, \ldots}$} for trace variables and \mbox{$Ti = \set{y_1, y_2, \ldots}$} to indicate time.
	Since we are now in a branching-time setting, we translate the quantification of paths with \mple second-order quantification.
\end{proof}

\begin{thm}\label{MPLEHyperCTLStar}
	\mple is strictly more expressive than HyperCTL$^*$, even if restricted to bisimulation-invariant properties.
\end{thm}
\begin{proof}
	We proceed by arguing that in contrast to HyperCTL$^*$, \mple can simulate the epistemic knowledge modality known from KLTL and KCTL$^*$~\cite{fagin2004reasoning, halpern1989}.
	Epistemic logics describe multiple agents and the local knowledge they gain from observing different aspects of the system. 
	Consider the extension of HyperCTL$^*$ with the modality $\mathcal{K}_{A, \pi} \varphi$ (HyperKCTL$^*$), stating that the agent who can observe only the propositions $A \subseteq AP$ on path $\pi$ knows that $\varphi$ holds. 
	The modality has the following semantics:
	\begin{alignat*}{3}
	&\Pi,i \models_\mathcal{T} \K_{A,\pi} \varphi &&\text{ iff } &&\forall p \in \mathit{Paths}(\mathcal{T}).~ p[0,i] =_A \Pi(\pi)[0,i] \\
	& && && \rightarrow \Pi[\pi \mapsto p], i \models_{\mathcal T} \varphi.
	\end{alignat*}
	Here, $=_A$ is used to state that two paths are $A$-observationally equivalent, i.e., equal when projected to the set $A$. 
	One can extend the translation from HyperCTL$^*$ to \mple
	to also handle $\mathcal{K}_{A, \pi} \varphi$. 
	The idea is to quantify a new second-order set for variable $\pi$ which agrees with the previously quantified set on all atomic propositions in $A$. 
	This shows that \mple subsumes HyperCTL$^*$ extended with the knowledge operator. 
	Note that due to the branching-time character of HyperKCTL$^*$, all properties expressible in that logic must be bisimulation invariant. 
	Thus, also the fragment of \mple that is equivalent to HyperKCTL$^*$ can only characterize bisimulation-invariant properties.
	It is, however, known that HyperCTL$^*$ does not subsume KCTL$^*$~\cite{Bozzelli}, which is a syntactic subset of HyperKCTL$^*$.
	Hence, \mple must be strictly more expressive than HyperCTL$^*$.
	
\end{proof}

We now show that HyperQCTL$^*$ is more expressive than \mple. 
This result coincides with the fact that QCTL$^*$ is known to be more expressive than MPL, which is equivalent to CTL$^*$~\cite{CTLStar-SAT,moller1999}.

\begin{defn}[HyperQCTL$^*$]
	\label{hyperqctl}
HyperQCTL$^*$ extends HyperCTL$^*$ (Def.~\ref{hyperctl}) with quantification over atomic propositions. 
We add atomic formulas $q_\pi$ and propositional quantification $\exists q. \varphi$ to the syntax of HyperCTL$^*$ to obtain HyperQCTL$^*$.
Let $\mathcal{T}$ be an $AP$-labeled tree with labeling function $L$.
The satisfaction relation of HyperQCTL$^*$ extends the one of HyperCTL$^*$ by adding a rule for the propositional quantification to the definition of the semantics:
\begin{alignat*}{3}
&\Pi, i \models_\mathcal{T} \exists q. \varphi &&\text{ iff } && \exists L': S \rightarrow 2^{AP \cup \set{q}}.~ \forall s \in S. \\
& && && L'(s) =_{AP \setminus \set{q}} L(s) \land \Pi, i \models_{\mathcal{T}[L' / L]} \varphi.
\end{alignat*}
Here, $L'$ is the updated labeling function which (re)assigns the atomic proposition $q$ to nodes in $\mathcal{T}$. 
We write $\mathcal{T}[L' / L]$ for the tree $\mathcal{T}$ labeled with elements from $AP \cup \set{q}$ with the new labeling function $L'$.
Note that compared to the definition of HyperQPTL, we do not just add an independent $q$-trace to the model. Instead, $q$ is (re-)assigned at every node in the tree.
\end{defn}
We say that a tree $\mathcal{T}$ satisfies a HyperQCTL$^*$ formula $\varphi$, written $\mathcal{T} \models \varphi$, if $\emptyset, 0 \models_\mathcal{T} \varphi$.
\begin{thm}\label{HyperQCTLStarMPLE}
	HyperQCTL$^*$ is strictly more expressive than \mple.
\end{thm}
\begin{proof}
	The proof proceeds as the one for Theorem~\ref{HyperQPTLmoreExpressiveFOE}.
	Consider the model of \emph{linear trees}, i.e., trees in which each node has a unique successor. 
	For this class of models, \mple is equivalent to \fo, since the equal-level predicate collapses to equality and second-order quantification in \mple can only quantify the unique single path. 
	Likewise, HyperQCTL$^*$ collapses to QPTL. 
	But it is known that QPTL can express all $\omega$-regular properties~\cite{sistla1987complementation} and \fo, which is equivalent to LTL, cannot~\cite{mcnaughton1971}.
	Thus, HyperQCTL$^*$ has to be more expressive than \mple. 
\end{proof}

Lastly, we prove the equivalence of \msoe and HyperQCTL$^*$.
The syntax and semantics of \msoe is the same as for \mple.
The only difference is that the second-order quantification is not restricted to range over full paths.

\begin{thm}
	\label{thm:HyperQCTL*equivMSOE}
	HyperQCTL$^*$ and \msoe are expressively equivalent.
\end{thm}

\begin{proof}
	We give linear translations in both directions.
	To translate an \msoe formula into HyperQCTL$^*$, we can use propositional quantification to mimic second-order quantification of a set of nodes in the tree. 
	Furthermore, to translate first-order quantification into HyperQCTL$^*$, we quantify a path and an atomic proposition, of which we require that it holds exactly once on that path (as in Lemma~\ref{HyperQPTLsubsumesFOE}). 
	This allows us to translate the equal-level predicate $E(x,y)$ by checking whether the atomic proposition for $x$ holds at the same time on the path for $x$ as the atomic proposition for $y$ on the path for $y$.
	For the other direction, we translate a  HyperQCTL$^*$ formula into an \msoe formula by using distinct first-order variables to indicate time and paths (as in Lemma~\ref{S1SESubsumesHyperQPTL}). 
	Quantification of atomic propositions can be mimicked by second-order quantification.
\end{proof}

\section{Satisfiability}
\label{sec:satisfiability}

The satisfiability problem of a standard linear-time logic is to decide,
for a given formula $\varphi$,  whether there exists a trace $t$ such that $t \models \varphi$. For a linear-time hyperlogic, we ask whether a trace set $T$
exists such that  $T  \models \varphi$. For branching-time logics, both for standard and for hyperlogics, we ask whether there is a tree $\mathcal{T}$ such that $\mathcal{T} \models \varphi$.

The satisfiability problems of LTL and CTL$^*$ are decidable~\cite{LTL-SAT, CTLStar-SAT}. In the linear-time spectrum, it is known that the fragment of $\exists^* \forall^*$~HyperLTL is decidable, and that already a single $\forall \exists$~trace quantifier alternation leads to undecidability~\cite{DecidingHyperLTL}. 
We show in Section~\ref{sec:HyperQPTLSat} that the same rule also applies to the more expressive linear-time hyperlogic HyperQPTL: the $\exists^* \forall^*$~HyperQPTL fragment is decidable whereas a $\forall \exists$~trace quantifier alternation causes undecidability.

In the branching-time spectrum, no decidable hyperlogics have been identified so far. 
In Section~\ref{sec:HyperCTLStarSat},
we first consider the satisfiability problem of the existential fragment of HyperCTL$^*$. We then obtain the general result that the  $\exists^* \forall^*$ fragment is decidable and that a single $\forall\exists$ quantifier alternation leads to undecidability.

\subsection{Linear-Time Satisfiability}
\label{sec:HyperQPTLSat}

We define the fragments of HyperQPTL on the basis of the formulas \emph{trace} quantification (analogously to HyperLTL~\cite{DecidingHyperLTL}).
Note that formulas in these fragments may contain arbitrary propositional quantification.

\begin{lem}
	\label{lem:hyperqptl-sat-forall}
The satisfiability problem of HyperQPTL formulas in the $\forall^*$ fragment is decidable.
\end{lem}

\begin{proof}
We employ a similar reasoning as in~\cite{DecidingHyperLTL}.
First note that a $\forall^*$ HyperQPTL formula is satisfiable iff it is satisfiable by a model containing only a single trace. Since all quantifiers are universal, a model with more than one trace can always be reduced to a model with exactly one trace.
We therefore reduce the satisfiability problem of $\forall^*$~HyperQPTL to the satisfiability problem of QPTL, which is decidable~\cite{QPTL}.
Given a $\forall^*$~HyperQPTL formula $\varphi$, consider the QPTL formula $\psi$ obtained by removing the trace quantifiers from $\varphi$ and also removing all trace variables from atomic propositions $a_\pi$.
The resulting formula is equisatisfiable to $\varphi$.
\end{proof}

\begin{thm}
	\label{hyperqptl-sat-decidable}
The satisfiability problem of HyperQPTL formulas in the $\exists^* \forall^*$~fragment is decidable.
\end{thm}

\begin{proof}
Let a HyperQPTL formula $\psi$ of the form $\vec{Q}_0.~\exists \pi_0, \ldots, \pi_n.~\vec{Q}_1.~\forall \pi'_0, \ldots, \pi'_m.~\vec{Q}_2.~ \varphi$ over $\mathit{AP} = \{a^0,\ldots,a^{k-1}\}$ be given.
$\vec{Q}_i$ denotes arbitrary propositional quantification.
In the following, we assume that $\psi$ has at least one existential trace quantifier.
If this is not the case, Lemma~\ref{lem:hyperqptl-sat-forall} handles the formula. 
We construct an equisatisfiable QPTL formula $\psi_\mathit{QPTL}$ for $\psi$.
As first step, we eliminate the universal quantification by explicitly enumerating every possible interaction between the universal and existential quantifiers, a technique already used to prove the decidability of the $\exists^* \forall^*$~HyperLTL fragment~\cite{DecidingHyperLTL}:
$$
\vec{Q}_0\exists \pi_0,\ldots, \pi_n. \vec{Q}_1 \bigwedge_{j_1 = 1}^n \hspace{-1.5mm}\ldots \hspace{-1.5mm} \bigwedge_{j_m = 1}^n \hspace{-1.5mm} \vec{Q}_2.\varphi[\pi_{j_1} / \pi'_1, \ldots, \pi_{j_m} / \pi'_m],
$$
where $\varphi[\pi_{j} / \pi'_i]$ denotes that the trace variable $\pi'_i$ in $\varphi$ is replaced by $\pi_j$. 
Like this, every combination of trace assignments for the universal quantification is covered and the resulting formula is equisatisfiable and of size $\mathcal{O}(n^m)$. 
The formula is technically not a HyperQPTL formula yet, since $\vec{Q}_2$ is in the scope of a conjunction. 
An equisatisfiable HyperQPTL formula without $\vec{Q}_2$ can be constructed by introducing a fresh existential quantifier in the quantifier prefix for each existential quantification in a conjunct (denoted by $\vec{\exists}$) and by simply moving the universal quantification in $\vec{Q}_2$ in front of the conjunction, into the quantifier prefix (denoted by $\vec{\forall}$).
We use $\varphi'$ to denote the body of the resulting HyperQPTL formula.
As second step, we eliminate the existential trace quantification by replacing each trace quantifier $\exists \pi_i$ with $k$ existential quantifiers over propositions, one for every atomic proposition in $\mathit{AP}$:
$$
\psi_\mathit{QPTL} \coloneqq \vec{Q}_0\bigexists_{\pi = \pi_0}^{\pi_n} a_\pi^0, \ldots, a_\pi^{k-1} \vec{Q}_1 \vec{\exists}~ \vec{\forall}.~ \varphi'.
$$
By the semantics of HyperQPTL, the resulting formula is equisatisfiable to $\psi$:
Assume $\psi$ is satisfied by a trace set $T_\psi$.
We construct a set of witnesses for $\psi_\mathit{QPTL}$ by splitting every witness $t \in T_\psi$ into $k$ traces $t^j$ ($j < k$), one for every atomic proposition $a^j \in \mathit{AP}$. 
For all such $j$, we require the traces $t^j \in (2^{\set{a_\pi^j}})^\omega$ to agree with $t$ on the translated atomic proposition.
The resulting trace set $\{t^j \mid 0 \leq j < k \}$ satisfies $\psi_\mathit{QPTL}$ by construction.
Constructing $T_\psi$ from a set of witnesses for a formula $\psi_\mathit{QPTL}$ obtained from the above construction works analogously by computing, for every point in time $k$, the union of all the witness positions $t^j[k]$ of witnesses for $\psi_\mathit{QPTL}$ for each $\pi_i$.
\end{proof}

\begin{thm}
\label{thm:HyperQPTLUndec}
	The satisfiability problem for HyperQPTL formulas in the $\forall\exists$ fragment is undecidable.
\end{thm}
\begin{proof}
	HyperQPTL subsumes HyperLTL as a syntactic fragment. Since the satisfiability problem for HyperLTL is undecidable for formulas with at least one trace quantifier alternation starting with a universal quantifier~\cite{DecidingHyperLTL}, this also holds for HyperQPTL.
\end{proof}

We have thus identified  the $\exists^* \forall^*$~fragment as the largest decidable fragment of the linear-time hyperlogic HyperQPTL.

\subsection{Branching-Time Satisfiability}
\label{sec:HyperCTLStarSat}
In the following, we study the satisfiability problem of HyperCTL$^*$. 
We assume formulas to be given in negated normal form.
A HyperCTL$^*$ formula in NNF is in the $\exists^*$ and $\forall^*$~fragment, respectively, if it contains exclusively universal or exclusively existential path quantifiers. 
The union of the two fragments is the alternation-free fragment.  
The formula is in the $\exists^* \forall^*$~fragment, if there is no existential path quantifier in the scope of a universal path quantifier. 
It is in the $\forall\exists$~fragment if there is exactly one existential path quantifier in the scope of a single universal path quantifier. 
\begin{lem}
	\label{lem:HyperCTL-SAT-forall}
	The satisfiability problem of HyperCTL$^*$ formulas in the $\forall^*$~fragment is decidable.
\end{lem}
\begin{proof}
	The proof is similar to the proof for Lemma~\ref{lem:hyperqptl-sat-forall}.
	We reduce the satisfiability problem of $\forall^*$~HyperCTL$^*$ to the satisfiability problem of CTL$^*$ and consider models that are linear trees, i.e., trees having only a single path. 
\end{proof}
We now prove in two steps that the $\exists^*$ fragment is decidable, regardless of the temporal modalities and the nesting depth of the existential quantifiers. We start with formulas of a specific form and then generalize the result to the full fragment. 
Subsequently, we establish that the satisfiability problem for HyperCTL$^*$ formulas in the full $\exists^* \forall^*$ fragment remains decidable. 

\begin{lem}\label{lem:HyperCTL-SAT-Release}
	The satisfiability problem for HyperCTL$^*$ formulas of the form $\varphi\coloneqq \exists \pi.(\exists \pi'. \psi') \LTLrelease (\exists \pi''. \psi'')$, where $\psi'$ and $\psi''$ are quantifier free, is decidable.
\end{lem}

\begin{proof}
	The key idea of the proof is to show that every model of a $\varphi$-shaped formula has a finite representation.
	More concretely, we show that we can represent an arbitrary model $\mathcal{T}$ satisfying $\varphi$ as a tree $\mathcal{T}_{\mathit{fin}}$ of bounded size. 
	We then show that $\mathcal{T}_{\mathit{fin}}$ can be extended to an infinite tree $\mathcal{\tilde T}$ which satisfies $\varphi$.
	We conclude by describing a naive decision procedure which enumerates all bounded trees $\mathcal{T_{\mathit{fin}}}$ and checks whether they can be extended into an infinite model $\mathcal{\tilde T}$ for $\varphi$.
	
	\textbf{Intuition.}
	We first give some intuition on how to construct $\mathcal{T}_{\mathit{fin}}$ out of $\mathcal{T}$.
	Assume a formula $\exists \pi. \LTLglobally (\exists \pi''. \psi)$ (which belongs to the fragment described in the statement) and a model $\mathcal{T}$ satisfying it.
	In this model, there needs to be a path $p$ witnessing $\pi$, and at each point in time $i$, there must be a path $p_i$, which branches off of $p$ and serves as a witness for $\pi''$ at point in time $i$.
	Extracting these witnesses from the model results in a comb-like structure as depicted in Figure~\ref{intro1}.
	
	\begin{figure}[t]
		\centering
		{\scalebox{0.9}{
				\begin{tikzpicture}
					\draw [line width = 0.4mm] (0,1.5) -- (0,5);
					\draw [line width = 0.4mm] (0,5) node (root0) {} -- (6.5,5);
					\draw [line width = 0.4mm] (0,4) node (root1) {} -- (6.5,4);
					\draw [line width = 0.4mm] (0,3) node (root2) {} -- (6.5,3);
					\draw [line width = 0.4mm] (0,2) node (root3) {} -- (6.5,2);

					\draw [draw = none, fill = black] (root0) circle [radius=0.7mm, label=left:A] node[label=left:{$p[0]$}, label=north east:{$p_0[0]$}] (dot00) {};
					\draw [draw = none, fill = black] (root0) [right = 2cm] circle [radius=0.7mm] node[label=above:{$p_0[1]$}] (dot01) {};
					\draw [draw = none, fill = black] (root0) [right = 4cm] circle [radius=0.7mm] node[label=above:{$p_0[2]$}] (dot02) {};
					\draw [draw = none, fill = black] (root0) [right = 6cm] circle [radius=0.7mm] node[label=above:{$p_0[3]$}] (dot03) {};
					
					\draw [draw = none, fill = black] (root1) circle [radius=0.7mm] node[label=left:{$p[1]$}, label=north east:{$p_1[0]$}] (dot10) {};
					\draw [draw = none, fill = black] (root1) [right = 2cm] circle [radius=0.7mm] node[label=above:{$p_1[1]$}] (dot11) {};
					\draw [draw = none, fill = black] (root1) [right = 4cm] circle [radius=0.7mm] node[label=above:{$p_1[2]$}] (dot12) {};
					\draw [draw = none, fill = black] (root1) [right = 6cm] circle [radius=0.7mm] node[label=above:{$p_1[3]$}] (dot13) {};
					
					\draw [draw = none, fill = black] (root2) circle [radius=0.7mm] node[label=left:{$p[2]$}, label=north east:{$p_2[0]$}] (dot20) {};
					\draw [draw = none, fill = black] (root2) [right = 2cm] circle [radius=0.7mm] node[label=above:{$p_2[1]$}] (dot21) {};
					\draw [draw = none, fill = black] (root2) [right = 4cm] circle [radius=0.7mm] node[label=above:{$p_2[2]$}] (dot22) {};
					\draw [draw = none, fill = black] (root2) [right = 6cm] circle [radius=0.7mm] node[label=above:{$p_2[3]$}] (dot23) {};
					
					\draw [draw = none, fill = black] (root3) circle [radius=0.7mm] node[label=left:{$p[3]$}, label=north east:{$p_3[0]$}, label=south west:{\color{darkgray!90}$D_3$}] (dot30) {};
					\draw [draw = none, fill = black] (root3) [right = 2cm] circle [radius=0.7mm] node[label=above:{$p_3[1]$}] (dot31) {};
					\draw [draw = none, fill = black] (root3) [right = 4cm] circle [radius=0.7mm] node[label=above:{$p_3[2]$}] (dot32) {};
					\draw [draw = none, fill = black] (root3) [right = 6cm] circle [radius=0.7mm] node[label=above:{$p_3[3]$}] (dot33) {};
					
					\draw [draw = none] (dot03) node [right=1cm] (tau0) {$p_0$};
					\draw [draw = none] (dot13) node [right=1cm] (tau1) {$p_1$};
					\draw [draw = none] (dot23) node [right=1cm] (tau2) {$p_2$};
					\draw [draw = none] (dot33) node [right=1cm] (tau3) {$p_3$};
					\draw [draw = none] (root3) node [below = 1cm, label=west:{$p$}] (tau) {};
					
					\draw [line width = 0.3mm, dotted] (dot03) -- (tau0);
					\draw [line width = 0.3mm, dotted] (dot13) -- (tau1);
					\draw [line width = 0.3mm, dotted] (dot23) -- (tau2);
					\draw [line width = 0.3mm, dotted] (dot33) -- (tau3);
					\draw [line width = 0.3mm, dotted] (root3) -- (tau);
					\draw (root3) [below = 0.8cm, right = 2.5cm] node {$\mathbb{\ldots}$};
					
					\draw [line width = 0.9mm, color=gray, opacity=0.5] ($(dot30)+(-2.2mm, -1.1mm)$) -- (dot03.north east);
					
				\end{tikzpicture}
		}}
		\caption{The witness $p$ for $\pi$ and the sequence of witnesses $p_i$ for $\pi''$ arranged in a comb-like structure. Nodes $p_3[0], p_2[1], p_1[2]$, and $p_0[3]$ reside on diagonal $D_3$.}
		\label{intro1}
	\end{figure}
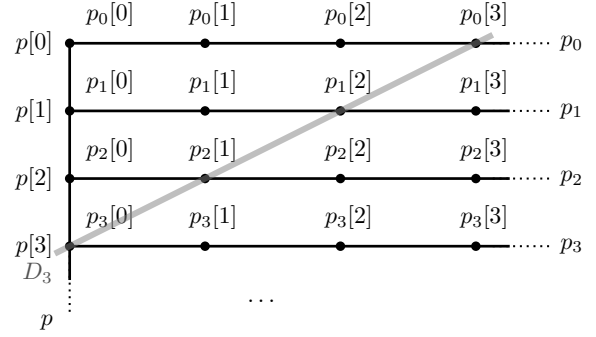
	
	Through formula $\psi$, each pair of nodes $p[i+j]$ and $p_i[j]$ are related with each other, e.g., if $\psi = \LTLglobally(a_\pi \leftrightarrow a_{\pi'})$, then all $p[i+j]$ and $p_i[j]$ must agree on $a$.
	The nodes $p[i+j]$ and $p_i[j]$ always reside on the same diagonal in the comb.
	Like this, all nodes on the same diagonal are related with each other through $p$.	
	When transforming the witnesses, it is therefore important to only consider one diagonal as a whole and to not alter just a single node on it.
	Diagonal $D_3$ is also depicted in Figure~\ref{intro1}.
	Next, note that $\psi$ can be transformed into a B\"uchi automaton which accepts each pair of witness paths $(p[i,\infty], p_i)$.
	We label each $p_i$ in the comb with the corresponding accepting automaton run.
	Now, the crucial observation is that if two diagonals in the comb are labeled with the same set of automaton states, then we can just \emph{cut out} the part between those two diagonals and still have accepting runs, i.e., the resulting paths $p$ and all $p_i$ are still witnesses for $\pi$ and $\pi''$.
	The proof proceeds by repeatedly cutting out nonessential parts of the comb until it has a suitable prefix of bounded size which we call $\mathcal{T}_{\mathit{fin}}$.
	
	\textbf{Formal Proof.}	
	We assume w.l.o.g. that no $\X$-modality occurs in the formula; any $\X$-modality can only add an offset to the operations which we describe in the following.
	Assume a tree $\mathcal{T}$ over nodes $S$ with labeling function $L$ that satisfies $\varphi$, i.e., there exist a path $p$ through $\mathcal{T}$ serving as witness for $\pi$.
	By the semantics of the $\R$-modality, either (Case~1) an infinite sequence $p_0, p_1, p_2, \ldots$ of witnesses for $\pi''$ or (Case~2) a finite sequence $p_0, p_1, \ldots, p_n$ of witnesses for $\pi''$ and a final witness $p'$ for $\pi'$.

	\emph{Case~1.}
	We first construct a B\"uchi automaton $A_{\psi''}$ that accepts all possible pairs of paths $(p,p_i)$ that satisfy $\psi''$. Then, we formally describe the comb structure and how it can be labeled with the accepting runs of $A_{\psi''}$ on each $(p[i,\infty],p_i)$. Furthermore, we give the necessary definitions of \emph{frontiers} and \emph{cuts} in order to define how to cut parts out of the comb without changing the fact that its paths constitute witnesses for $\varphi$. We show how to construct $\mathcal{T}_{\mathit{fin}}$ by repeatedly cutting out nonessential parts of the comb and give a maximal bound for its size. Lastly, we describe how to extend $\mathcal{T}_{\mathit{fin}}$ to an infinite model $\mathcal{\tilde T}$.
	
	\emph{Automaton construction.}
	Since the formula $\psi''$ is quantifier-free, it is an LTL formula where each atomic proposition $a_\pi$ is interpreted as a unique LTL proposition.
	Let $A'_{\psi''} = (Q', q'_0, \Sigma', \delta', F')$ be the nondeterministic B\"uchi automaton obtained from this LTL interpretation of $\psi''$~\cite{LTLAutomata}.
	We transform $A'_{\psi''}$ into a nondeterministic B\"uchi automaton $A_{\psi''} = (Q,q_0,\Sigma,\delta,F)$, which reasons separately over $\pi$ and $\pi''$:
	\begin{itemize}
		\item $Q: Q'$, $q_0: q'_0$
		
		\item $\Sigma: S \times S$
		
		\item $\delta: Q \times \Sigma \rightarrow 2^Q$ where $q' \in \delta(q, (s_0, s_1))$ iff $q' \in \delta'(q, A)$ and $L(s_0) = \{a \in AP ~|~ a_\pi \in A \} $ and $L(s_1) = \{a \in AP ~|~ a_{\pi''} \in A \} $, where $A \subseteq \{ a_\pi, a_{\pi''} ~|~ a \in AP\}.$ 
		
		\item $F: F'$
		
	\end{itemize}
Note that $A_{\psi''}$ reasons over pairs of paths, while $A'_{\psi''}$ reasons over traces.
The automaton $A_{\psi''}$ yields accepting runs $r_i$ for all pairs of witnesses $(p[i,\infty],p_i)$.
We can thus \emph{associate} with each node $p_i[j]$ the automaton state $r_i[j]$.

\emph{Frontiers.}
We arrange the witness paths $p$ and $p_0, p_1, p_2, \ldots$ in a comb-like structure as shown in Figure~\ref{intro1}.
For all $k \in \nat$, we address the sequence of nodes $p_0[k-0], \ldots , p_k[k-k]$ as the $k$-th diagonal of the comb, denoted by $D_k$.
We use the usual sequence notation for diagonals, e.g., $D_k[i]$ to address the $i$-th element of the sequence.
We call the set of automaton states associated with nodes in $D_k$ \emph{frontier} $F_k$, formally $F_k \coloneqq \{r_i[k-i] ~|~ i \leq k\}$.
Note that for every $k$, $F_k \subseteq Q$.

\emph{Cuts.}
We now establish how to safely remove parts of the comb structure, i.e., in such a way that the witness paths in the altered comb still have accepting runs in $A_{\psi''}$. To this end, we define the \emph{cut operation} and refine it to a \emph{preserving cut}, which requires the definition of an additional property which we call \emph{frontier-preserving cuttable}.

For two diagonals $D_k$ and $D_{k'}$ with $F_k = F_{k'}$, a cut modifies the comb in such a way that the suffix of every node in $D_k$ is replaced by the suffix of a node in $D_{k'}$, where both nodes have to be associated with the same automaton state. 
Formally, we replace every $p_i[k-i, \infty]$ with some $p_{i'}[k'-i', \infty]$, requiring that $D_k[i]$ and $D_{k'}[i']$ are associated with the same state.
Additionally, to preserve the relation of the modified paths with $p$, we replace the sub-comb with origin in $p[k]$ with the sub-comb with origin in $p[k']$.
Note that because of the requirement that $D_k[i]$ and $D_{k'}[i']$ are associated with the same state, the modified witness paths still have accepting runs through $A_{\psi''}$.

For $k \leq k'$, we say that two diagonals $D_k$, $D_{k'} $ with $F_k = F_{k'}$ are frontier-preserving cuttable (for short: cuttable) if for every $q \in F_{k'}$,
$q$ is either associated with at least as many nodes on $D_k$ as on $D_{k'}$, or it is associated with $|Q|$ nodes on $D_k$.

For $k \leq k' \leq k''$, a cut preserving $F_{k''}$ is a cut between two cuttable diagonals $D_k$ and $D_{k'}$, such that the set $F_{k''}$ is not modified by the operation.
For each $q \in F_{k''}$, pick a position $i_{q} \leq k''$ as a representative such that the state associated with $D_{k''}[i_{q}]$ is $q$.
All states $q$ with representative position $i_q \geq k'$ will not be affected by the cut.
For representative positions $i_q < k'$, ensure that when choosing suffices from $D_{k'}$ for the cut, each suffix $p_{i_q}[k'-i_q, \infty]$ is chosen at least once.
This is possible since we require $D_k$ and $D_{k'}$ to be cuttable.
Like this, we ensured that all representative states are not deleted by the cut.
Figure~\ref{cutting1} and Figure~\ref{cutting2} show the choice of representative positions in a comb and the resulting preserving cut.

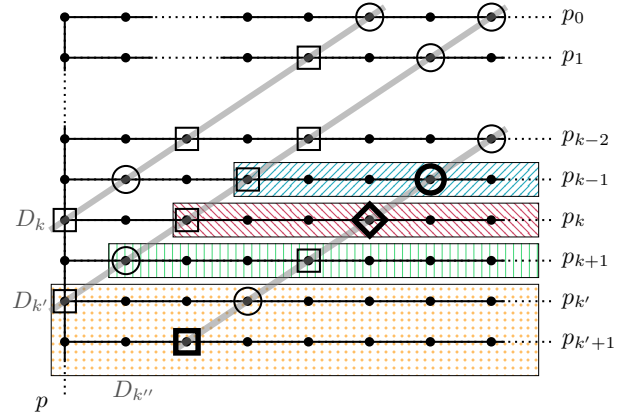
\begin{figure}[t]
	\centering
	{\scalebox{0.9}{
			\begin{tikzpicture}
			\pgfdeclarelayer{bg}    
			\pgfsetlayers{bg,main}
			
			\draw [line width = 0.3mm] (0,0.7) -- (0,4.2);
			\draw [line width = 0.3mm] (0,5) -- (0,5.8);
			
			\draw [line width = 0.3mm] (0,5.8) node (root0) {} -- (1.3,5.8);
			\draw [line width = 0.3mm] (2.3,5.8) -- (6.5,5.8);
			
			\draw [line width = 0.3mm] (0,5.2) node (root1) {} -- (1.3,5.2);
			\draw [line width = 0.3mm] (2.3,5.2) -- (6.5,5.2);
			
			\draw [line width = 0.3mm] (0,4) node (root2) {} -- (6.5,4);
			\draw [line width = 0.3mm] (0,3.4) node (root3) {} -- (6.5,3.4);
			\draw [line width = 0.3mm] (0,2.8) node (root4) {} -- (6.5,2.8);
			\draw [line width = 0.3mm] (0,2.2) node (root5) {} -- (6.5,2.2);
			\draw [line width = 0.3mm] (0,1.6) node (root6) {} -- (6.5,1.6);
			\draw [line width = 0.3mm] (0,1) node (root7) {} -- (6.5,1);
			
			\draw [draw = none, fill = black] (root0) circle [radius=0.7mm] node (dot00) {};
			\draw [draw = none, fill = black] (root0) [right = .9cm] circle [radius=0.7mm] node (dot01) {};
			\draw [draw = none, fill = black] (root0) [right = 2.7cm] circle [radius=0.7mm] node (dot03) {};
			\draw [draw = none, fill = black] (root0) [right = 3.6cm] circle [radius=0.7mm] node (dot04) {};
			\draw [draw = none, fill = black] (root0) [right = 4.5cm] circle [radius=0.7mm] node (dot05) {};
			\draw [draw = none, fill = black] (root0) [right = 5.4cm] circle [radius=0.7mm] node (dot06) {};
			\draw [draw = none, fill = black] (root0) [right = 6.3cm] circle [radius=0.7mm] node (dot07) {};
			
			\draw [draw = none, fill = black] (root1) circle [radius=0.7mm] node (dot10) {};
			\draw [draw = none, fill = black] (root1) [right = .9cm] circle [radius=0.7mm] node (dot11) {};
			\draw [draw = none, fill = black] (root1) [right = 2.7cm] circle [radius=0.7mm] node (dot13) {};
			\draw [draw = none, fill = black] (root1) [right = 3.6cm] circle [radius=0.7mm] node (dot14) {};
			\draw [draw = none, fill = black] (root1) [right = 4.5cm] circle [radius=0.7mm] node (dot15) {};
			\draw [draw = none, fill = black] (root1) [right = 5.4cm] circle [radius=0.7mm] node (dot16) {};
			\draw [draw = none, fill = black] (root1) [right = 6.3cm] circle [radius=0.7mm] node (dot17) {};
			
			\draw [draw = none, fill = black] (root2) circle [radius=0.7mm] node (dot20) {};
			\draw [draw = none, fill = black] (root2) [right = .9cm] circle [radius=0.7mm] node (dot21) {};
			\draw [draw = none, fill = black] (root2) [right = 1.8cm] circle [radius=0.7mm] node (dot22) {};
			\draw [draw = none, fill = black] (root2) [right = 2.7cm] circle [radius=0.7mm] node (dot23) {};
			\draw [draw = none, fill = black] (root2) [right = 3.6cm] circle [radius=0.7mm] node (dot24) {};
			\draw [draw = none, fill = black] (root2) [right = 4.5cm] circle [radius=0.7mm] node (dot25) {};
			\draw [draw = none, fill = black] (root2) [right = 5.4cm] circle [radius=0.7mm] node (dot26) {};
			\draw [draw = none, fill = black] (root2) [right = 6.3cm] circle [radius=0.7mm] node (dot27) {};
			
			\draw [draw = none, fill = black] (root3) circle [radius=0.7mm] node (dot30) {};
			\draw [draw = none, fill = black] (root3) [right = .9cm] circle [radius=0.7mm] node (dot31) {};
			\draw [draw = none, fill = black] (root3) [right = 1.8cm] circle [radius=0.7mm] node (dot32) {};
			\draw [draw = none, fill = black] (root3) [right = 2.7cm] circle [radius=0.7mm] node (dot33) {};
			\draw [draw = none, fill = black] (root3) [right = 3.6cm] circle [radius=0.7mm] node (dot34) {};
			\draw [draw = none, fill = black] (root3) [right = 4.5cm] circle [radius=0.7mm] node (dot35) {};
			\draw [draw = none, fill = black] (root3) [right = 5.4cm] circle [radius=0.7mm] node (dot36) {};
			\draw [draw = none, fill = black] (root3) [right = 6.3cm] circle [radius=0.7mm] node (dot37) {};
			
			\draw [draw = none, fill = black] (root4) circle [radius=0.7mm] node (dot40) {};
			\draw [draw = none, fill = black] (root4) [right = .9cm] circle [radius=0.7mm] node (dot41) {};
			\draw [draw = none, fill = black] (root4) [right = 1.8cm] circle [radius=0.7mm] node (dot42) {};
			\draw [draw = none, fill = black] (root4) [right = 2.7cm] circle [radius=0.7mm] node (dot43) {};
			\draw [draw = none, fill = black] (root4) [right = 3.6cm] circle [radius=0.7mm] node (dot44) {};
			\draw [draw = none, fill = black] (root4) [right = 4.5cm] circle [radius=0.7mm] node (dot45) {};
			\draw [draw = none, fill = black] (root4) [right = 5.4cm] circle [radius=0.7mm] node (dot46) {};
			\draw [draw = none, fill = black] (root4) [right = 6.3cm] circle [radius=0.7mm] node (dot47) {};
			
			\draw [draw = none, fill = black] (root5) circle [radius=0.7mm] node (dot50) {};
			\draw [draw = none, fill = black] (root5) [right = .9cm] circle [radius=0.7mm] node (dot51) {};
			\draw [draw = none, fill = black] (root5) [right = 1.8cm] circle [radius=0.7mm] node (dot52) {};
			\draw [draw = none, fill = black] (root5) [right = 2.7cm] circle [radius=0.7mm] node (dot53) {};
			\draw [draw = none, fill = black] (root5) [right = 3.6cm] circle [radius=0.7mm] node (dot54) {};
			\draw [draw = none, fill = black] (root5) [right = 4.5cm] circle [radius=0.7mm] node (dot55) {};
			\draw [draw = none, fill = black] (root5) [right = 5.4cm] circle [radius=0.7mm] node (dot56) {};
			\draw [draw = none, fill = black] (root5) [right = 6.3cm] circle [radius=0.7mm] node (dot57) {};
			
			\draw [draw = none, fill = black] (root6) circle [radius=0.7mm] node (dot60) {};
			\draw [draw = none, fill = black] (root6) [right = .9cm] circle [radius=0.7mm] node (dot61) {};
			\draw [draw = none, fill = black] (root6) [right = 1.8cm] circle [radius=0.7mm] node (dot62) {};
			\draw [draw = none, fill = black] (root6) [right = 2.7cm] circle [radius=0.7mm] node (dot63) {};
			\draw [draw = none, fill = black] (root6) [right = 3.6cm] circle [radius=0.7mm] node (dot64) {};
			\draw [draw = none, fill = black] (root6) [right = 4.5cm] circle [radius=0.7mm] node (dot65) {};
			\draw [draw = none, fill = black] (root6) [right = 5.4cm] circle [radius=0.7mm] node (dot66) {};
			\draw [draw = none, fill = black] (root6) [right = 6.3cm] circle [radius=0.7mm] node (dot67) {};
			
			\draw [draw = none, fill = black] (root7) circle [radius=0.7mm] node (dot70) {};
			\draw [draw = none, fill = black] (root7) [right = .9cm] circle [radius=0.7mm] node (dot71) {};
			\draw [draw = none, fill = black] (root7) [right = 1.8cm] circle [radius=0.7mm] node (dot72) {};
			\draw [draw = none, fill = black] (root7) [right = 2.7cm] circle [radius=0.7mm] node (dot73) {};
			\draw [draw = none, fill = black] (root7) [right = 3.6cm] circle [radius=0.7mm] node (dot74) {};
			\draw [draw = none, fill = black] (root7) [right = 4.5cm] circle [radius=0.7mm] node (dot75) {};
			\draw [draw = none, fill = black] (root7) [right = 5.4cm] circle [radius=0.7mm] node (dot76) {};
			\draw [draw = none, fill = black] (root7) [right = 6.3cm] circle [radius=0.7mm] node (dot77) {};
						
			\draw [draw = none] (dot07) node [right = 0.8cm] (tau0) {$p_0$};
			\draw [draw = none] (dot17) node [right = 0.8cm] (tau1) {$p_1$};
			\draw [draw = none] (dot27) node [right = 0.8cm] (tau2) {$p_{k-2}$};
			\draw [draw = none] (dot37) node [right = 0.8cm] (tau3) {$p_{k-1}$};
			\draw [draw = none] (dot47) node [right = 0.8cm] (tau4) {$p_{k}$};
			\draw [draw = none] (dot57) node [right = 0.8cm] (tau5) {$p_{k+1}$};
			\draw [draw = none] (dot67) node [right = 0.8cm] (tau6) {$p_{k'}$};
			\draw [draw = none] (dot77) node [right = 0.8cm] (tau7) {$p_{k'+1}$};
			\draw [draw = none] (dot70) node [below = 0.8cm, label=west:{$p$}] (tau) {};
			
			\draw [line width = 0.3mm, dotted] (root0) -- (tau0);
			\draw [line width = 0.3mm, dotted] (root1) -- (tau1);
			\draw [line width = 0.3mm, dotted] (root2) -- (tau2);
			\draw [line width = 0.3mm, dotted] (root3) -- (tau3);
			\draw [line width = 0.3mm, dotted] (root4) -- (tau4);
			\draw [line width = 0.3mm, dotted] (root5) -- (tau5);
			\draw [line width = 0.3mm, dotted] (root6) -- (tau6);
			\draw [line width = 0.3mm, dotted] (root7) -- (tau7);
			\draw [line width = 0.3mm, dotted] (root0) -- (tau);
			
			
			\draw [line width = 0.9mm, color=gray, opacity=0.5] ($(dot40)+(-2.2mm, -1.5mm)$) -- ($(dot05)+(1mm, 1.5mm)$) ;
			\draw [draw=none] (dot40) node [xshift = -0.5cm] (F3) {\color{darkgray!90}$D_k$};
			
			\draw [line width = 0.9mm, color=gray, opacity=0.5] ($(dot60)+(-2.2mm, -1.5mm)$) -- ($(dot07)+(1mm, 1.5mm)$) ;
			\draw [draw=none] (dot60) node [xshift = -0.5cm] (F3) {\color{darkgray!90}$D_{k'}$};
			
			\draw [line width = 1mm, color=gray, opacity=0.5] ($(dot72)+(-2.8mm, -1mm)$) -- ($(dot27)+(1mm, 1.5mm)$);
			\draw [draw=none] (dot71) node [below = 4.5mm] (F3) {\color{darkgray!90}$D_{k''}$};
			
			\draw [line width = 0.3mm] (dot05) circle [radius=2mm, xshift = -1.2mm] node (acc10) {};
			\draw [line width = 0.3mm] (dot31) circle [radius=2mm, xshift = -1.2mm] node (acc10) {};
			\node[draw,fit=(dot40),line width = 0.3mm, inner sep = 0.4mm] (acc01) {};
			\node[draw,fit=(dot22),line width = 0.3mm, inner sep = 0.4mm, xshift = -1.2mm] (acc01) {};
			\node[draw,fit=(dot14),line width = 0.3mm, inner sep = 0.4mm, xshift = -1.2mm] (acc01) {};
			
			\draw [line width = 0.3mm] (dot51) circle [radius=2mm, xshift = -1.2mm] node (acc10) {};
			\draw [line width = 0.3mm] (dot16) circle [radius=2mm, xshift = -1.2mm] node (acc10) {};
			\draw [line width = 0.3mm] (dot07) circle [radius=2mm, xshift = -1.2mm] node (acc10) {};
			\node[draw,fit=(dot60),line width = 0.3mm, inner sep = 0.4mm] (acc01) {};
			\node[draw,fit=(dot42),line width = 0.3mm, inner sep = 0.4mm, xshift = -1.2mm] (acc01) {};
			\node[draw,fit=(dot33),line width = 0.3mm, inner sep = 0.4mm, xshift = -1.2mm] (acc01) {};
			\node[draw,fit=(dot24),line width = 0.3mm, inner sep = 0.4mm, xshift = -1.2mm] (acc01) {};
			
			\draw [line width = 0.3mm] (dot63) circle [radius=2mm, xshift = -1.2mm] node (acc10) {};
			\draw [line width = 0.8mm] (dot36) circle [radius=2mm, xshift = -1.2mm] node (acc10) {};
			\draw [line width = 0.3mm] (dot27) circle [radius=2mm, xshift = -1.2mm] node (acc10) {};
			
			\node[draw,fit=(dot72),line width = 0.8mm, inner sep = 0.4mm, xshift = -1.2mm] (acc01) {};
			\node[draw,fit=(dot54),line width = 0.3mm, inner sep = 0.4mm, xshift = -1.2mm] (acc01) {};
			
			\node[draw,fit=(dot45),line width = 0.8mm, inner sep = 0.4mm, xshift = -1.2mm, rotate=45] (acc01) {};
			
			\begin{pgfonlayer}{bg}
				\draw[pattern=north east lines, step = 2cm, pattern color = MidnightBlue!60] (2.5,3.65) rectangle (7,3.15);
				\draw[pattern=north west lines, pattern color = Maroon!60] (1.6,3.05) rectangle (7,2.55);
				\draw[pattern=vertical lines, pattern color = PineGreen!60] (0.65,2.45) rectangle (7,1.95);
				\draw[pattern=dots, pattern color = YellowOrange!60] (-.2,0.5) rectangle (7,1.85);
		 	\end{pgfonlayer}
			\end{tikzpicture}
	}}
	\caption{A comb structure with highlighted diagonals $D_k$, $D_{k'}$, and $D_{k''}$, together
		with their associated automaton states (depicted as square, circle and diamond).
		States at representative positions are printed in bold, suffices used for the cut are highlighted.}
	\label{cutting1}
\end{figure}
\begin{figure}[t]
	\centering
	{\scalebox{0.9}{
			\begin{tikzpicture}
			\pgfdeclarelayer{bg}    
			\pgfsetlayers{bg,main}
			
			\draw [line width = 0.3mm] (0,0.7) -- (0,4.2);
			\draw [line width = 0.3mm] (0,5) -- (0,5.8);
			
			\draw [line width = 0.3mm] (0,5.8) node (root0) {} -- (1.3,5.8);
			\draw [line width = 0.3mm] (2.3,5.8) -- (6.5,5.8);
			
			\draw [line width = 0.3mm] (0,5.2) node (root1) {} -- (1.3,5.2);
			\draw [line width = 0.3mm] (2.3,5.2) -- (6.5,5.2);
			
			\draw [line width = 0.3mm] (0,4) node (root2) {} -- (6.5,4);
			\draw [line width = 0.3mm] (0,3.4) node (root3) {} -- (6.5,3.4);
			\draw [line width = 0.3mm] (0,2.8) node (root4) {} -- (6.5,2.8);
			\draw [line width = 0.3mm] (0,2.2) node (root5) {} -- (6.5,2.2);
			\draw [line width = 0.3mm] (0,1.6) node (root6) {} -- (6.5,1.6);
			\draw [line width = 0.3mm] (0,1) node (root7) {} -- (6.5,1);
			
			\draw [draw = none, fill = black] (root0) circle [radius=0.7mm] node (dot00) {};
			\draw [draw = none, fill = black] (root0) [right = .9cm] circle [radius=0.7mm] node (dot01) {};
			\draw [draw = none, fill = black] (root0) [right = 2.7cm] circle [radius=0.7mm] node (dot03) {};
			\draw [draw = none, fill = black] (root0) [right = 3.6cm] circle [radius=0.7mm] node (dot04) {};
			\draw [draw = none, fill = black] (root0) [right = 4.5cm] circle [radius=0.7mm] node (dot05) {};
			\draw [draw = none, fill = black] (root0) [right = 5.4cm] circle [radius=0.7mm] node (dot06) {};
			\draw [draw = none, fill = black] (root0) [right = 6.3cm] circle [radius=0.7mm] node (dot07) {};
			
			\draw [draw = none, fill = black] (root1) circle [radius=0.7mm] node (dot10) {};
			\draw [draw = none, fill = black] (root1) [right = .9cm] circle [radius=0.7mm] node (dot11) {};
			\draw [draw = none, fill = black] (root1) [right = 2.7cm] circle [radius=0.7mm] node (dot13) {};
			\draw [draw = none, fill = black] (root1) [right = 3.6cm] circle [radius=0.7mm] node (dot14) {};
			\draw [draw = none, fill = black] (root1) [right = 4.5cm] circle [radius=0.7mm] node (dot15) {};
			\draw [draw = none, fill = black] (root1) [right = 5.4cm] circle [radius=0.7mm] node (dot16) {};
			\draw [draw = none, fill = black] (root1) [right = 6.3cm] circle [radius=0.7mm] node (dot17) {};
			
			\draw [draw = none, fill = black] (root2) circle [radius=0.7mm] node (dot20) {};
			\draw [draw = none, fill = black] (root2) [right = .9cm] circle [radius=0.7mm] node (dot21) {};
			\draw [draw = none, fill = black] (root2) [right = 1.8cm] circle [radius=0.7mm] node (dot22) {};
			\draw [draw = none, fill = black] (root2) [right = 2.7cm] circle [radius=0.7mm] node (dot23) {};
			\draw [draw = none, fill = black] (root2) [right = 3.6cm] circle [radius=0.7mm] node (dot24) {};
			\draw [draw = none, fill = black] (root2) [right = 4.5cm] circle [radius=0.7mm] node (dot25) {};
			\draw [draw = none, fill = black] (root2) [right = 5.4cm] circle [radius=0.7mm] node (dot26) {};
			\draw [draw = none, fill = black] (root2) [right = 6.3cm] circle [radius=0.7mm] node (dot27) {};
			
			\draw [draw = none, fill = black] (root3) circle [radius=0.7mm] node (dot30) {};
			\draw [draw = none, fill = black] (root3) [right = .9cm] circle [radius=0.7mm] node (dot31) {};
			\draw [draw = none, fill = black] (root3) [right = 1.8cm] circle [radius=0.7mm] node (dot32) {};
			\draw [draw = none, fill = black] (root3) [right = 2.7cm] circle [radius=0.7mm] node (dot33) {};
			\draw [draw = none, fill = black] (root3) [right = 3.6cm] circle [radius=0.7mm] node (dot34) {};
			\draw [draw = none, fill = black] (root3) [right = 4.5cm] circle [radius=0.7mm] node (dot35) {};
			\draw [draw = none, fill = black] (root3) [right = 5.4cm] circle [radius=0.7mm] node (dot36) {};
			\draw [draw = none, fill = black] (root3) [right = 6.3cm] circle [radius=0.7mm] node (dot37) {};
			
			\draw [draw = none, fill = black] (root4) circle [radius=0.7mm] node (dot40) {};
			\draw [draw = none, fill = black] (root4) [right = .9cm] circle [radius=0.7mm] node (dot41) {};
			\draw [draw = none, fill = black] (root4) [right = 1.8cm] circle [radius=0.7mm] node (dot42) {};
			\draw [draw = none, fill = black] (root4) [right = 2.7cm] circle [radius=0.7mm] node (dot43) {};
			\draw [draw = none, fill = black] (root4) [right = 3.6cm] circle [radius=0.7mm] node (dot44) {};
			\draw [draw = none, fill = black] (root4) [right = 4.5cm] circle [radius=0.7mm] node (dot45) {};
			\draw [draw = none, fill = black] (root4) [right = 5.4cm] circle [radius=0.7mm] node (dot46) {};
			\draw [draw = none, fill = black] (root4) [right = 6.3cm] circle [radius=0.7mm] node (dot47) {};
			
			\draw [draw = none, fill = black] (root5) circle [radius=0.7mm] node (dot50) {};
			\draw [draw = none, fill = black] (root5) [right = .9cm] circle [radius=0.7mm] node (dot51) {};
			\draw [draw = none, fill = black] (root5) [right = 1.8cm] circle [radius=0.7mm] node (dot52) {};
			\draw [draw = none, fill = black] (root5) [right = 2.7cm] circle [radius=0.7mm] node (dot53) {};
			\draw [draw = none, fill = black] (root5) [right = 3.6cm] circle [radius=0.7mm] node (dot54) {};
			\draw [draw = none, fill = black] (root5) [right = 4.5cm] circle [radius=0.7mm] node (dot55) {};
			\draw [draw = none, fill = black] (root5) [right = 5.4cm] circle [radius=0.7mm] node (dot56) {};
			\draw [draw = none, fill = black] (root5) [right = 6.3cm] circle [radius=0.7mm] node (dot57) {};
			
			\draw [draw = none, fill = black] (root6) circle [radius=0.7mm] node (dot60) {};
			\draw [draw = none, fill = black] (root6) [right = .9cm] circle [radius=0.7mm] node (dot61) {};
			\draw [draw = none, fill = black] (root6) [right = 1.8cm] circle [radius=0.7mm] node (dot62) {};
			\draw [draw = none, fill = black] (root6) [right = 2.7cm] circle [radius=0.7mm] node (dot63) {};
			\draw [draw = none, fill = black] (root6) [right = 3.6cm] circle [radius=0.7mm] node (dot64) {};
			\draw [draw = none, fill = black] (root6) [right = 4.5cm] circle [radius=0.7mm] node (dot65) {};
			\draw [draw = none, fill = black] (root6) [right = 5.4cm] circle [radius=0.7mm] node (dot66) {};
			\draw [draw = none, fill = black] (root6) [right = 6.3cm] circle [radius=0.7mm] node (dot67) {};
			
			\draw [draw = none, fill = black] (root7) circle [radius=0.7mm] node (dot70) {};
			\draw [draw = none, fill = black] (root7) [right = .9cm] circle [radius=0.7mm] node (dot71) {};
			\draw [draw = none, fill = black] (root7) [right = 1.8cm] circle [radius=0.7mm] node (dot72) {};
			\draw [draw = none, fill = black] (root7) [right = 2.7cm] circle [radius=0.7mm] node (dot73) {};
			\draw [draw = none, fill = black] (root7) [right = 3.6cm] circle [radius=0.7mm] node (dot74) {};
			\draw [draw = none, fill = black] (root7) [right = 4.5cm] circle [radius=0.7mm] node (dot75) {};
			\draw [draw = none, fill = black] (root7) [right = 5.4cm] circle [radius=0.7mm] node (dot76) {};
			\draw [draw = none, fill = black] (root7) [right = 6.3cm] circle [radius=0.7mm] node (dot77) {};
						
			\draw [draw = none] (dot07) node [right = 0.8cm] (tau0) {$p_0$};
			\draw [draw = none] (dot17) node [right = 0.8cm] (tau1) {$p_1$};
			\draw [draw = none] (dot27) node [right = 0.8cm] (tau2) {$p_{k-2}$};
			\draw [draw = none] (dot37) node [right = 0.8cm] (tau3) {$p_{k-1}$};
			\draw [draw = none] (dot47) node [right = 0.8cm] (tau4) {$p_{k}$};
			\draw [draw = none] (dot57) node [right = 0.8cm] (tau5) {$p_{k+1}$};
			\draw [draw = none] (dot67) node [right = 0.8cm] (tau6) {$p_{k'}$};
			\draw [draw = none] (dot77) node [right = 0.8cm] (tau7) {$p_{k'+1}$};
			\draw [draw = none] (dot70) node [below = 0.8cm, label=west:{$p$}] (tau) {};
			
			\draw [line width = 0.3mm, dotted] (root0) -- (tau0);
			\draw [line width = 0.3mm, dotted] (root1) -- (tau1);
			\draw [line width = 0.3mm, dotted] (root2) -- (tau2);
			\draw [line width = 0.3mm, dotted] (root3) -- (tau3);
			\draw [line width = 0.3mm, dotted] (root4) -- (tau4);
			\draw [line width = 0.3mm, dotted] (root5) -- (tau5);
			\draw [line width = 0.3mm, dotted] (root6) -- (tau6);
			\draw [line width = 0.3mm, dotted] (root7) -- (tau7);
			\draw [line width = 0.3mm, dotted] (root0) -- (tau);
			
			\draw [line width = 0.9mm, color=gray, opacity=0.5] ($(dot40)+(-2.2mm, -1.5mm)$) -- ($(dot05)+(1mm, 1.5mm)$) ;
			\draw [draw=none] (dot40) node [xshift = -0.5cm] (F3) {\color{darkgray!90}$D_{k'}$};
			
			\draw [line width = 0.9mm, color=gray, opacity=0.5] ($(dot70)+(-2.2mm, -1.5mm)$) -- ($(dot17)+(1mm, 1.5mm)$) ;
			\draw [draw=none] (dot70) node [xshift = -0.5cm] (F3) {\color{darkgray!90}$D_{k''}$};
									
			\draw [line width = 0.3mm] (dot31) circle [radius=2mm, xshift = -1.2mm] node (acc10) {};
			\draw [line width = 0.3mm] (dot05) circle [radius=2mm, xshift = -1.2mm] node (acc10) {};
			\node[draw,fit=(dot40),line width = 0.3mm, inner sep = 0.4mm] (acc01) {};
			\node[draw,fit=(dot22),line width = 0.3mm, inner sep = 0.4mm, xshift = -1.2mm] (acc01) {};
			\node[draw,fit=(dot14),line width = 0.3mm, inner sep = 0.4mm, xshift = -1.2mm] (acc01) {};
			
			\draw [line width = 0.3mm] (dot43) circle [radius=2mm, xshift = -1.2mm] node (acc10) {};
			\draw [line width = 0.8mm] (dot25) circle [radius=2mm, xshift = -1.2mm] node (acc10) {};
			
			\node[draw,fit=(dot61),line width = 0.3mm, inner sep = 0.4mm, xshift = -1.2mm] (acc01) {};
			\node[draw,fit=(dot70),line width = 0.3mm, inner sep = 0.4mm] (acc01) {};
			\node[draw,fit=(dot34),line width = 0.3mm, inner sep = 0.4mm, xshift = -1.2mm] (acc01) {};
			\node[draw,fit=(dot52),line width = 0.8mm, inner sep = 0.4mm, xshift = -1.2mm] (acc01) {};
			
			\node[draw,fit=(dot17),line width = 0.8mm, inner sep = 0.4mm, xshift = -1.2mm, rotate=45] (acc01) {};
			
			\begin{pgfonlayer}{bg}	
				\draw[pattern=dots, pattern color = YellowOrange!60] (-.2,0.5) rectangle (7,3.05);
				
				\draw[pattern=vertical lines, pattern color = PineGreen!60] (4.25,5.55) rectangle (7,6.05);
				\draw[pattern=north west lines, pattern color = Maroon!60] (3.35,4.95) rectangle (7,5.45);
				\draw[pattern=north east lines, step = 2cm, pattern color = MidnightBlue!60] (1.55,3.75) rectangle (7,4.25);			
				\draw[pattern=vertical lines, pattern color = PineGreen!60] (0.65,3.15) rectangle (7,3.65);
		 	\end{pgfonlayer}
			\end{tikzpicture}
	}}
	\caption{The result of the cutting operation prepared in Figure~\ref{cutting1}. 
		The highlights show which suffix was shifted to which node in the comb.}
	\label{cutting2}
\end{figure}
\emph{Construct $\mathcal{T_{\mathit{fin}}}$.}
We describe how to perform a series of preserving cuts to ensure that sufficiently many accepting states can be found in a bounded-size prefix $\mathcal{T_{\mathit{fin}}}$ of the comb.
First, note that there are at most $2^{|Q|}$ different frontiers.
Furthermore, there are at most $c = (|Q|+1)^{|Q|}$ many different equivalence classes of the cuttable property, i.e. for $c+1$ many diagonals, at least two are cuttable.
We say that a diagonal $D_{k}$ is \emph{close} to $D_{k'}$ if $|k' - k| \leq c$.
By the pigeonhole principle, for every two diagonals $D_k, D_{k'}$ and state set $F_{k''}$ with $k \leq k' \leq k''$, we can perform a number of cuts on diagonals situated between $D_k$ and $D_{k'}$, each preserving $F_{k''}$, such that at the end, $D_{k'}$ is close to $D_k$ and the set $F_{k''}$ did not change.

There are only finitely many different frontiers in the infinite comb, so at least one frontier occurs on infinitely many diagonals. We call that frontier $F_\omega$.
Pick the smallest number $\mathit{inf} \in \nat$ such that $\qinf = F_\omega$ and cut diagonal $\finf$ as close as possible to $D_0$ while preserving $F_{\mathit{inf}}$.
Note that the first $|Q|^{|Q|}$ diagonals are, in general, not cuttable; therefore, in the worst case, $D_\mathit{inf}$ will be cut close to $D_{|Q|^{|Q|}}$.
As a result of these cuts, frontier $F_\omega$ might not occur infinitely often anymore.
More concretely, diagonals which were previously associated with frontier $F_\omega$ will now have frontiers which are a subset of $F_\omega$.
Since there are only finitely many different subsets of any finite set, we know that there exists at least one frontier $F_{\omega'} \subseteq F_\omega$ that occurs on infinitely many diagonals.

For every automaton state $q \in F_{\omega'}$, there exists by construction an $i_q \leq \mathit{inf}$ such that $\finf[i_q]$ is associated with $q$.
We call the set of all $i_q$ the set of designated positions~$P$. 
Now, find the smallest $\mathit{inf'} > \mathit{inf}$, such that $F_{\mathit{inf'}} = F_{\omega'}$, and for all $i \in P$, $r_i$ has an accepting state between $\mathit{inf}$ and $\mathit{inf'}$.
Such an $\mathit{inf'}$ exist because of the B\"uchi acceptance condition.
We now perform a series of cuts to cut $\finff$ as close to $\finf$ as possible, each of which preserves $\qinff$.
Find the $i \in P$ whose accepting state is closest to $\finf$ and cut the corresponding diagonal close to $\finf$.
Continue with the $i' \in P$ whose accepting state comes next and cut it close to the last diagonal that was cut close.
Proceed, until the diagonal of the last accepting state of designated position was cut close.
Finally, cut $\finff$ close to that last diagonal.

We choose $\mathcal{T_{\mathit{fin}}}$ to be the finite prefix of the resulting comb up to (and including) $\finff$.
The depth of $\mathcal{T_{\mathit{fin}}}$ is bounded by 
$b = |Q|^{|Q|} 
+ (2 + |Q|) \cdot (|Q| + 1)^{|Q|} 
$.
This is because $\mathcal{T_{\mathit{fin}}}$ consists of a prefix of diagonals up to the first cuttable diagonal (at the most $|Q|^{|Q|}$ many), followed by $D_\mathit{inf}$. 
Then, $|P| \leq |Q| $ many diagonals have been cut close and the distance between them is at the most $(|Q| + 1)^{|Q|}$. 
Lastly, $\finff$ was cut close, again with a maximal distance of $(|Q| + 1)^{|Q|}$.

\begin{figure}[t]
	\centering
	{\scalebox{0.9}{
			\begin{tikzpicture}
			\pgfdeclarelayer{bg}    
			\pgfsetlayers{bg,main}
			
			\draw [line width = 0.3mm] (0,1.6) -- (0,4.2);
			\draw [line width = 0.3mm] (0,5) -- (0,5.8);
			
			\draw [line width = 0.3mm] (0,5.8) node (root0) {} -- (1.3,5.8);
			\draw [line width = 0.3mm] (2.3,5.8) -- (6.3,5.8);
			
			\draw [line width = 0.3mm] (0,5.2) node (root1) {} -- (1.3,5.2);
			\draw [line width = 0.3mm] (2.3,5.2) -- (5.4,5.2);
			
			\draw [line width = 0.3mm] (0,4) node (root2) {} -- (3.6,4);
			\draw [line width = 0.3mm] (0,3.4) node (root3) {} -- (2.7,3.4);
			\draw [line width = 0.3mm] (0,2.8) node (root4) {} -- (1.8,2.8);
			\draw [line width = 0.3mm] (0,2.2) node (root5) {} -- (0.9,2.2);
			
			\draw [draw = none, fill = black] (root0) circle [radius=0.7mm] node (dot00) {};
			\draw [draw = none, fill = black] (root0) [right = .9cm] circle [radius=0.7mm] node (dot01) {};
			\draw [draw = none, fill = black] (root0) [right = 2.7cm] circle [radius=0.7mm] node (dot03) {};
			\draw [draw = none, fill = black] (root0) [right = 3.6cm] circle [radius=0.7mm] node (dot04) {};
			\draw [draw = none, fill = black] (root0) [right = 4.5cm] circle [radius=0.7mm] node (dot05) {};
			\draw [draw = none, fill = black] (root0) [right = 5.4cm] circle [radius=0.7mm] node (dot06) {};
			\draw [draw = none, fill = black] (root0) [right = 6.3cm] circle [radius=0.7mm] node (dot07) {};
			
			\draw [draw = none, fill = black] (root1) circle [radius=0.7mm] node (dot10) {};
			\draw [draw = none, fill = black] (root1) [right = .9cm] circle [radius=0.7mm] node (dot11) {};
			\draw [draw = none, fill = black] (root1) [right = 2.7cm] circle [radius=0.7mm] node (dot13) {};
			\draw [draw = none, fill = black] (root1) [right = 3.6cm] circle [radius=0.7mm] node (dot14) {};
			\draw [draw = none, fill = black] (root1) [right = 4.5cm] circle [radius=0.7mm] node (dot15) {};
			\draw [draw = none, fill = black] (root1) [right = 5.4cm] circle [radius=0.7mm] node (dot16) {};
			
			\draw [draw = none, fill = black] (root2) circle [radius=0.7mm] node (dot20) {};
			\draw [draw = none, fill = black] (root2) [right = .9cm] circle [radius=0.7mm] node (dot21) {};
			\draw [draw = none, fill = black] (root2) [right = 1.8cm] circle [radius=0.7mm] node (dot22) {};
			\draw [draw = none, fill = black] (root2) [right = 2.7cm] circle [radius=0.7mm] node (dot23) {};
			\draw [draw = none, fill = black] (root2) [right = 3.6cm] circle [radius=0.7mm] node (dot24) {};
			
			\draw [draw = none, fill = black] (root3) circle [radius=0.7mm] node (dot30) {};
			\draw [draw = none, fill = black] (root3) [right = .9cm] circle [radius=0.7mm] node (dot31) {};
			\draw [draw = none, fill = black] (root3) [right = 1.8cm] circle [radius=0.7mm] node (dot32) {};
			\draw [draw = none, fill = black] (root3) [right = 2.7cm] circle [radius=0.7mm] node (dot33) {};
			
			\draw [draw = none, fill = black] (root4) circle [radius=0.7mm] node (dot40) {};
			\draw [draw = none, fill = black] (root4) [right = .9cm] circle [radius=0.7mm] node (dot41) {};
			\draw [draw = none, fill = black] (root4) [right = 1.8cm] circle [radius=0.7mm] node (dot42) {};
			
			\draw [draw = none, fill = black] (root5) circle [radius=0.7mm] node (dot50) {};
			\draw [draw = none, fill = black] (root5) [right = .9cm] circle [radius=0.7mm] node (dot51) {};
			
			\draw [draw = none, fill = black] (root6) circle [radius=0.7mm] node (dot60) {};
		
			\draw [line width = 0.3mm, dotted] (root0) -- (dot03);
			\draw [line width = 0.3mm, dotted] (root1) -- (dot13);
			\draw [line width = 0.3mm, dotted] (root0) -- (dot30);
			
			\draw [line width = 0.9mm, color=gray, opacity=0.5] ($(dot40)+(-2.2mm, -1.5mm)$) -- ($(dot05)+(1mm, 1.5mm)$) ;
			\draw [draw=none] (dot40) node [xshift = -0.7cm] (F3) {\color{darkgray!90}$D_{\mathit{inf}}$};
			
			\draw [line width = 0.9mm, color=gray, opacity=0.5] ($(dot60)+(-2.2mm, -1.5mm)$) -- ($(dot07)+(1mm, 1.5mm)$) ;
			\draw [draw=none] (dot60) node [xshift = -0.7cm] (F3) {\color{darkgray!90}$D_{\mathit{inf'}}$};
			
			\draw [line width = 0.3mm] (dot05) circle [radius=2mm, xshift = -1.2mm] node (acc10) {};
			\draw [line width = 0.3mm] (dot22) circle [radius=2mm, xshift = -1.2mm] node (acc10) {};
			\node[draw,fit=(dot40),line width = 0.8mm, inner sep = 0.4mm] (acc01) {};
			\node[draw,fit=(dot31),line width = 0.3mm, inner sep = 0.4mm, xshift = -1.2mm] (acc01) {};
			\node[draw,fit=(dot14),line width = 0.8mm, inner sep = 0.4mm, xshift = -1.2mm, rotate=45] (acc01) {};
			
			\node[draw,fit=(dot07),line width = 0.3mm, inner sep = 0.4mm, xshift = -1.2mm] (acc01) {};
			\node[draw,fit=(dot16),line width = 0.3mm, inner sep = 0.4mm, xshift = -1.2mm] (acc01) {};
			\node[draw,fit=(dot60),line width = 0.3mm, inner sep = 0.4mm] (acc01) {};
			\node[draw,fit=(dot42),line width = 0.3mm, inner sep = 0.4mm, xshift = -1.2mm] (acc01) {};
			\node[draw,fit=(dot51),line width = 0.3mm, inner sep = 0.4mm, xshift = -1.2mm] (acc01) {};
			\node[draw,fit=(dot24),line width = 0.3mm, inner sep = 0.4mm, xshift = -1.2mm, rotate=45] (acc01) {};
			\node[draw,fit=(dot33),line width = 0.3mm, inner sep = 0.4mm, xshift = -1.2mm, rotate=45] (acc01) {};
			
			\begin{pgfonlayer}{bg}
				\draw[pattern=north east lines, step = 2cm, pattern color = MidnightBlue!60] (4.3,5.45) rectangle (5.8,4.95);
				\draw[pattern=vertical lines, pattern color = PineGreen!60] (0.65,3.05) rectangle (2.15,2.55);
				\draw[pattern=dots, pattern color = YellowOrange!60] (-.25,2.45) -- (1.8,2.45) -- (-.25,1) -- cycle;
		 	\end{pgfonlayer}
			\end{tikzpicture}
	}}
	\caption{The finite prefix $\mathcal{T_{\mathit{fin}}}$ with diagonals $\finf$ and $\finff$. States in designated positions on $\finf$ are printed in bold. Suffices that will be copied to extend the prefix comb are highlighted.}
	\label{pumping1}
\end{figure}
\begin{figure}[t]
	\centering
	{\scalebox{0.9}{
			\begin{tikzpicture}
			\pgfdeclarelayer{bg}    
			\pgfsetlayers{bg,main}
			
			\draw [line width = 0.3mm] (0,0.4) node (root8) {} -- (0,4.2);
			\draw [line width = 0.3mm] (0,5) -- (0,5.8);
			
			\draw [line width = 0.3mm] (0,5.8) node (root0) {} -- (1.3,5.8);
			\draw [line width = 0.3mm] (2.3,5.8) -- (8.1,5.8);
			
			\draw [line width = 0.3mm] (0,5.2) node (root1) {} -- (1.3,5.2);
			\draw [line width = 0.3mm] (2.3,5.2) -- (7.2,5.2);
			
			\draw [line width = 0.3mm] (0,4) node (root2) {} -- (5.4,4);
			\draw [line width = 0.3mm] (0,3.4) node (root3) {} -- (4.5,3.4);
			\draw [line width = 0.3mm] (0,2.8) node (root4) {} -- (3.6,2.8);
			\draw [line width = 0.3mm] (0,2.2) node (root5) {} -- (2.7,2.2);
			\draw [line width = 0.3mm] (0,1.6) node (root6) {} -- (1.8,1.6);
			\draw [line width = 0.3mm] (0,1) node (root7) {} -- (0.9,1);
			
			\draw [draw = none, fill = black] (root0) circle [radius=0.7mm] node (dot00) {};
			\draw [draw = none, fill = black] (root0) [right = .9cm] circle [radius=0.7mm] node (dot01) {};
			\draw [draw = none, fill = black] (root0) [right = 2.7cm] circle [radius=0.7mm] node (dot03) {};
			\draw [draw = none, fill = black] (root0) [right = 3.6cm] circle [radius=0.7mm] node (dot04) {};
			\draw [draw = none, fill = black] (root0) [right = 4.5cm] circle [radius=0.7mm] node (dot05) {};
			\draw [draw = none, fill = black] (root0) [right = 5.4cm] circle [radius=0.7mm] node (dot06) {};
			\draw [draw = none, fill = black] (root0) [right = 6.3cm] circle [radius=0.7mm] node (dot07) {};
			\draw [draw = none, fill = black] (root0) [right = 7.2cm] circle [radius=0.7mm] node (dot08) {};
			\draw [draw = none, fill = black] (root0) [right = 8.1cm] circle [radius=0.7mm] node (dot09) {};
			
			\draw [draw = none, fill = black] (root1) circle [radius=0.7mm] node (dot10) {};
			\draw [draw = none, fill = black] (root1) [right = .9cm] circle [radius=0.7mm] node (dot11) {};
			\draw [draw = none, fill = black] (root1) [right = 2.7cm] circle [radius=0.7mm] node (dot13) {};
			\draw [draw = none, fill = black] (root1) [right = 3.6cm] circle [radius=0.7mm] node (dot14) {};
			\draw [draw = none, fill = black] (root1) [right = 4.5cm] circle [radius=0.7mm] node (dot15) {};
			\draw [draw = none, fill = black] (root1) [right = 5.4cm] circle [radius=0.7mm] node (dot16) {};
			\draw [draw = none, fill = black] (root1) [right = 6.3cm] circle [radius=0.7mm] node (dot17) {};
			\draw [draw = none, fill = black] (root1) [right = 7.2cm] circle [radius=0.7mm] node (dot18) {};
			
			\draw [draw = none, fill = black] (root2) circle [radius=0.7mm] node (dot20) {};
			\draw [draw = none, fill = black] (root2) [right = .9cm] circle [radius=0.7mm] node (dot21) {};
			\draw [draw = none, fill = black] (root2) [right = 1.8cm] circle [radius=0.7mm] node (dot22) {};
			\draw [draw = none, fill = black] (root2) [right = 2.7cm] circle [radius=0.7mm] node (dot23) {};
			\draw [draw = none, fill = black] (root2) [right = 3.6cm] circle [radius=0.7mm] node (dot24) {};
			\draw [draw = none, fill = black] (root2) [right = 4.5cm] circle [radius=0.7mm] node (dot25) {};
			\draw [draw = none, fill = black] (root2) [right = 5.4cm] circle [radius=0.7mm] node (dot26) {};
			
			\draw [draw = none, fill = black] (root3) circle [radius=0.7mm] node (dot30) {};
			\draw [draw = none, fill = black] (root3) [right = .9cm] circle [radius=0.7mm] node (dot31) {};
			\draw [draw = none, fill = black] (root3) [right = 1.8cm] circle [radius=0.7mm] node (dot32) {};
			\draw [draw = none, fill = black] (root3) [right = 2.7cm] circle [radius=0.7mm] node (dot33) {};
			\draw [draw = none, fill = black] (root3) [right = 3.6cm] circle [radius=0.7mm] node (dot34) {};
			\draw [draw = none, fill = black] (root3) [right = 4.5cm] circle [radius=0.7mm] node (dot35) {};
			
			\draw [draw = none, fill = black] (root4) circle [radius=0.7mm] node (dot40) {};
			\draw [draw = none, fill = black] (root4) [right = .9cm] circle [radius=0.7mm] node (dot41) {};
			\draw [draw = none, fill = black] (root4) [right = 1.8cm] circle [radius=0.7mm] node (dot42) {};
			\draw [draw = none, fill = black] (root4) [right = 2.7cm] circle [radius=0.7mm] node (dot43) {};
			\draw [draw = none, fill = black] (root4) [right = 3.6cm] circle [radius=0.7mm] node (dot44) {};
			
			\draw [draw = none, fill = black] (root5) circle [radius=0.7mm] node (dot50) {};
			\draw [draw = none, fill = black] (root5) [right = .9cm] circle [radius=0.7mm] node (dot51) {};
			\draw [draw = none, fill = black] (root5) [right =1.8cm] circle [radius=0.7mm] node (dot52) {};
			\draw [draw = none, fill = black] (root5) [right =2.7cm] circle [radius=0.7mm] node (dot53) {};
			
			\draw [draw = none, fill = black] (root6) circle [radius=0.7mm] node (dot60) {};
			\draw [draw = none, fill = black] (root6) [right = .9cm] circle [radius=0.7mm] node (dot61) {};
			\draw [draw = none, fill = black] (root6) [right = 1.8cm] circle [radius=0.7mm] node (dot62) {};
			
			\draw [draw = none, fill = black] (root7) circle [radius=0.7mm] node (dot70) {};
			\draw [draw = none, fill = black] (root7) [right = .9cm] circle [radius=0.7mm] node (dot71) {};

			\draw [draw = none, fill = black] (root8) circle [radius=0.7mm] node (dot80) {};
					
			\draw [line width = 0.3mm, dotted] (root0) -- (dot03);
			\draw [line width = 0.3mm, dotted] (root1) -- (dot13);
			\draw [line width = 0.3mm, dotted] (root0) -- (dot30);
			
			\draw [line width = 0.9mm, color=gray, opacity=0.5] ($(dot40)+(-2.2mm, -1.5mm)$) -- ($(dot05)+(1mm, 1.5mm)$) ;
			\draw [draw=none] (dot40) node [xshift = -0.7cm] (F3) {\color{darkgray!90}$D_{\mathit{inf}}$};
			
			\draw [line width = 0.9mm, color=gray, opacity=0.5] ($(dot60)+(-2.2mm, -1.5mm)$) -- ($(dot07)+(1mm, 1.5mm)$) ;
			\draw [draw=none] (dot60) node [xshift = -0.7cm] (F3) {\color{darkgray!90}$D_{\mathit{inf'}}$};
			
			\draw [line width = 0.9mm, color=gray, opacity=0.5] ($(dot80)+(-2.2mm, -1.5mm)$) -- ($(dot09)+(1mm, 1.5mm)$) ;
			\draw [draw=none] (dot80) node [xshift = -0.7cm] (F3) {\color{darkgray!90}$D_{\mathit{inf''}}$};
			
			\draw [line width = 0.3mm] (dot05) circle [radius=2mm, xshift = -1.2mm] node (acc10) {};
			\draw [line width = 0.3mm] (dot22) circle [radius=2mm, xshift = -1.2mm] node (acc10) {};
			\node[draw,fit=(dot40),line width = 0.8mm, inner sep = 0.4mm] (acc01) {};
			\node[draw,fit=(dot31),line width = 0.3mm, inner sep = 0.4mm, xshift = -1.2mm] (acc01) {};
			\node[draw,fit=(dot14),line width = 0.8mm, inner sep = 0.4mm, xshift = -1.2mm, rotate=45] (acc01) {};
			
			\node[draw,fit=(dot07),line width = 0.3mm, inner sep = 0.4mm, xshift = -1.2mm] (acc01) {};
			\node[draw,fit=(dot16),line width = 0.3mm, inner sep = 0.4mm, xshift = -1.2mm] (acc01) {};
			\node[draw,fit=(dot60),line width = 0.3mm, inner sep = 0.4mm] (acc01) {};
			\node[draw,fit=(dot42),line width = 0.3mm, inner sep = 0.4mm, xshift = -1.2mm] (acc01) {};
			\node[draw,fit=(dot51),line width = 0.3mm, inner sep = 0.4mm, xshift = -1.2mm] (acc01) {};
			\node[draw,fit=(dot24),line width = 0.3mm, inner sep = 0.4mm, xshift = -1.2mm, rotate=45] (acc01) {};
			\node[draw,fit=(dot33),line width = 0.3mm, inner sep = 0.4mm, xshift = -1.2mm, rotate=45] (acc01) {};
			
			\node[draw,fit=(dot09),line width = 0.3mm, inner sep = 0.4mm, xshift = -1.2mm] (acc01) {};
			\node[draw,fit=(dot18),line width = 0.3mm, inner sep = 0.4mm, xshift = -1.2mm] (acc01) {};
			\node[draw,fit=(dot26),line width = 0.3mm, inner sep = 0.4mm, xshift = -1.2mm] (acc01) {};
			\node[draw,fit=(dot35),line width = 0.3mm, inner sep = 0.4mm, xshift = -1.2mm] (acc01) {};
			\node[draw,fit=(dot44),line width = 0.3mm, inner sep = 0.4mm, xshift = -1.2mm] (acc01) {};
			\node[draw,fit=(dot53),line width = 0.3mm, inner sep = 0.4mm, xshift = -1.2mm] (acc01) {};
			\node[draw,fit=(dot62),line width = 0.3mm, inner sep = 0.4mm, xshift = -1.2mm] (acc01) {};
			\node[draw,fit=(dot71),line width = 0.3mm, inner sep = 0.4mm, xshift = -1.2mm] (acc01) {};
			\node[draw,fit=(dot80),line width = 0.3mm, inner sep = 0.4mm] (acc01) {};
			
			\begin{pgfonlayer}{bg}
				\draw[pattern=north east lines, step = 2cm, pattern color = MidnightBlue!60] (4.25,4.25) rectangle (5.75,3.75);
				\draw[pattern=north east lines, step = 2cm, pattern color = MidnightBlue!60] (3.35,3.65) rectangle (4.85,3.15);
				\draw[pattern=vertical lines, pattern color = PineGreen!60] (6.95,6.05) rectangle (8.45,5.55);
				\draw[pattern=vertical lines, pattern color = PineGreen!60] (6.05,5.45) rectangle (7.55,4.95);
				\draw[pattern=vertical lines, pattern color = PineGreen!60] (2.45,3.05) rectangle (3.95,2.55);
				\draw[pattern=vertical lines, pattern color = PineGreen!60] (1.55,2.45) rectangle (3.05,1.95);
				\draw[pattern=vertical lines, pattern color = PineGreen!60] (0.65,1.85) rectangle (2.15,1.35);
				\draw[pattern=dots, pattern color = YellowOrange!60] (-.25,1.25) -- (1.8,1.25) -- (-.25,-.15) -- cycle;
		 	\end{pgfonlayer}
			\end{tikzpicture}
	}}
	\caption{The resulting larger finite prefix after extending every witness path in Figure~\ref{pumping1} once. The highlights show which part of $\mathcal{T_{\mathit{fin}}}$ was used to extend the prefix.}
	\label{pumping2}
\end{figure}

\emph{Construct $\mathcal{\tilde T}$.}
We now extend $\mathcal{T_{\mathit{fin}}}$ into an infinite tree $\mathcal{\tilde T}$ also satisfying $\varphi$. 
By construction, for each $i \in P$, run $r_i$ has an accepting state between $\finf$ and $\finff$.
Furthermore, for each $q \in \qinff$, there is a designated position $i_q \in P$ such that $\finf[i_q]$ is associated with $q$.
We extend $\mathcal{T_{\mathit{fin}}}$ by extending each node in $\finff$ as follows:
For each $i \leq \mathit{inf'}$ with $q$ associated to $\finff[i]$ and designated position $i_q$, we append a copy of $p_{i_q}[\mathit{inf}-i_q+1, \mathit{inf}' - i_q]$ to $p_i[\mathit{inf'} - i]$.
Additionally, we copy the sub-comb starting in node $p[\mathit{inf}+1]$ and append it to node $p[\mathit{inf'}]$, thus completing the extension.
By construction, we now have a larger finite comb ending in a diagonal $D_{\mathit{inf}''}$ with $F_{\mathit{inf}''} \subseteq F_{\mathit{inf}'}$. 
Figure~\ref{pumping1} shows a possible prefix comb $\mathcal{T_{\mathit{fin}}}$ and Figure~\ref{pumping2} shows how it is extended.
Repeating this process indefinitely, we get an infinite, ultimately periodic 
model $\mathcal{\tilde T}$ where each pair $(p[i, \infty], p_i)$ of witness paths in the comb of $\mathcal{\tilde T}$ is accepted by $A_{\psi''}$. It is thus a model for $\varphi$.

\emph{Case~2.}
In the case where the release modality is witnessed by path $p$ for $\pi$ and a sequence of paths $\{p_0, p_1, \ldots, p_n, p'\}$ for $\pi''$ and $\pi'$, we proceed very similar to Case~1. 
We again arrange the witnesses in a comb-like graph, with the only difference that at $p[n]$, there are the two witnesses $p_n$ and $p'$ branching from $p$.
In order to get the same structure as in Case~1, we zip $p'$ and $p_n$ into one witness path $\bar p_n$.
Furthermore, for all $m > n$, we add dummy witnesses $p_\top = \emptyset^\omega$ branching from $p$ at $p[m]$.

\emph{Automata construction.} As in Case~1, we associate the paths with the corresponding automaton runs.
For $(p[i, \infty], p_i)$ with $i < n$, we use the automaton $A_{\psi''}$, as in Case~1. 
For $(p[i, \infty], p_\top)$ with $i > n$, we use the automaton $A_\top$, which unconditionally accepts every pair of traces. 
For $(p[n, \infty], \bar p_n)$, we construct a new automaton $A_{\psi' \land \psi''}$ based on the LTL automaton $A'_{\psi' \land \psi''}$ for $\psi' \land \psi''$, similar to the construction of $A_{\psi''}$. 
The automaton $A_{\psi' \land \psi''}$ has the alphabet $\Sigma = S \times (S \times S)$ and the transition function $\delta: Q \times \Sigma \rightarrow 2^Q$, where $q' \in \delta(q, (s_0, (s_1,s_2)))$ iff $q' \in \delta'(q, A)$, and $L(s_0) = \{a \in AP ~|~ a_\pi \in A \} $, $L(s_1) = \{a \in AP ~|~ a_{\pi'} \in A \} $, and $L(s_2) = \{a \in AP ~|~ a_{\pi''} \in A \}$. We denote the set of states of automaton $A_{\psi''}$ with $Q$ and the set of states of automaton $A_{\psi' \land \psi''}$ with $Q_{\psi' \land \psi''}$. 

\emph{Construct $\mathcal{T_{\mathit{fin}}}$.}
Again, we construct a tree $\mathcal{T_{\mathit{fin}}}$ of bounded depth.
First, cut diagonal $D_n$ close to diagonal $D_0$.
Following $D_n$, there are again finitely many different cuttable diagonals (containing states from all three automata). 
Proceeding as in Case~1, construct $\mathcal{T_{\mathit{fin}}}$ such that after $D_n$, there are two diagonals $\finf$ and $\finff$ with sufficiently many accepting states in between.
The bound $b'$ on the depth of $\mathcal{T_{\mathit{fin}}}$ is obtained analogously to the bound in Case~1.
We only remark that $n$ is bounded by $|Q|^{|Q|} + (|Q| + 1)^{|Q|}$, and the maximal number of different cuttable diagonals is described in terms of the number of states of all three automata, i.e., $(|Q| + |Q_{\psi' \land \psi''}| + |Q_\top| + 1)^{(|Q| + |Q_{\psi' \land \psi''}| + |Q_\top|)}$. 
We conclude by noting that $b'$ can be used as an over-approximation of bound $b$ in Case~1.
Finally, as in Case~1, we construct an infinite satisfying tree $\mathcal{\tilde T}$ using $\mathcal{T_{\mathit{fin}}}$.

\emph{Decision Algorithm.}
Enumerate all comb-like prefixes $\mathcal{T_{\mathit{fin}}}$ of bounded depth $b'$ to find a suitable prefix (for either of both cases).
Whether a prefix is suitable or not can be decided by labeling it with corresponding runs from the automata of Cases 1 and 2 and checking whether it contains a segment between two diagonals $\finf$ and $\finff$ which qualifies to be extended into a model $\mathcal{\tilde T}$ as described above.
When associating the comb prefix with runs from the automata, we have to take into account all finitely-many points in time $n$ where $\psi''$ could be released by $\psi'$.
If some prefix $\mathcal{T_{\mathit{fin}}}$ is suitable, $\varphi$ is satisfiable (namely by the described tree $\mathcal{\tilde T}$).
As shown above, there is a suitable finite prefix of bounded depth $b'$ whenever $\varphi$ is satisfiable.
\end{proof}

\begin{cor}\label{cor:HyperCTL-SAT-Until}
	The satisfiability problem for HyperCTL$^*$ formulas of the form $\exists \pi.(\exists \pi''. \psi'') \LTLuntil (\exists \pi'. \psi')$, where $\psi'$ and $\psi''$ are quantifier free, is decidable.
\end{cor}

\begin{proof}
	We proceed similarly to Case~2 in the proof above. 
	The only difference is that formula $\psi''$ does not have to hold at the same point in time $n$ where formula $\psi'$ holds. 
	Therefore, the resulting comb does not have two witnesses branching from $p[n]$ that we have to zip. We use an automaton $A_{\psi'}$ instead of $A_{\psi' \land \psi''}$ to obtain the run $r_n$ for $(p[n, \infty], p')$. 
\end{proof}

We lift the arguments of the above proof to arbitrary formulas in the existential fragment of HyperCTL$^*$. 

\begin{lem}
	\label{hyperCTL$^*$-exists-decidabel}
	The satisfiability problem for HyperCTL$^*$ formulas in the $\exists^*$~fragment is decidable.
\end{lem}

\begin{proof}
	Define the existential quantifier depth of a $\exists ^*$HyperCTL$^*$ formula as the maximal number of alternations between existential quantifiers and the temporal modalities $\R$ and $\U$ in the syntax tree.
	The witnesses of a formula with quantifier depth $d$ can be arranged as a $d$-dimensional comb.
	We assume, again, w.l.o.g. that no $\X$-modality occurs in the formula.
	Lemma~\ref{lem:HyperCTL-SAT-Release} and Corollary~\ref{cor:HyperCTL-SAT-Until} cover the case where the comb is 2-dimensional.
	We now lift the arguments to the general case.
	Given a $d$-dimensional comb, we associate the innermost witnesses with the corresponding runs on the $d$-tuple automata, which we build as before from the inner LTL formulas.
	A $3$-dimensional comb and the corresponding automaton runs are exemplarily depicted in Figure~\ref{3dcomb}.

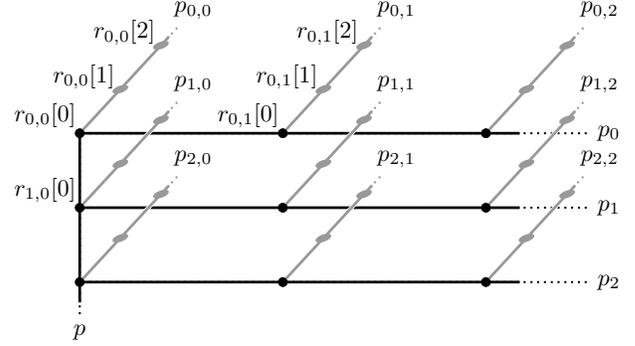
\begin{figure}[t]
	\centering
	{\scalebox{0.9}{
		\begin{tikzpicture}[x  = {(0cm,-1cm)}, 
							y  = {(1cm,0cm)},
							z  = {(0.6cm,0.65cm)},
							scale = 1]
			
			\begin{scope}[canvas is xy plane at z=0]
 				\draw[line width=0.4mm] (0,0) -- (2.5,0);
			
 				\draw[line width=0.4mm] (0,0) node (root0) {} -- (0,6.5);
				\draw[line width=0.4mm] (1.1,0) node (root1) {} -- (1.1,6.5);
				\draw[line width=0.4mm] (2.2,0) node (root2) {} -- (2.2,6.5);
				
				\draw [draw = none, fill = black] (root0) circle [radius=0.7mm] node (dot000) {};
				\draw [draw = none, fill = black] (root0) [right=3cm] circle [radius=0.7mm] node (dot010) {};
				\draw [draw = none, fill = black] (root0) [right=6cm] circle [radius=0.7mm] node (dot020) {};
				
				\draw [draw = none, fill = black] (root1) circle [radius=0.7mm] node (dot100) {};
				\draw [draw = none, fill = black] (root1) [right=3cm] circle [radius=0.7mm] node (dot110) {};
				\draw [draw = none, fill = black] (root1) [right=6cm] circle [radius=0.7mm] node (dot120) {};
				
				\draw [draw = none, fill = black] (root2) circle [radius=0.7mm] node (dot200) {};
				\draw [draw = none, fill = black] (root2) [right=3cm] circle [radius=0.7mm] node (dot210) {};
				\draw [draw = none, fill = black] (root2) [right=6cm] circle [radius=0.7mm] node (dot220) {};
				
				\draw [draw = none] (dot200) node [below = 0.5cm] (tau) {$p$};
				
				\draw [draw = none] (dot020) node [right = 1.4cm] (tau0) {$p_0$};
				\draw [draw = none] (dot120) node [right = 1.4cm] (tau1) {$p_1$};
				\draw [draw = none] (dot220) node [right = 1.4cm] (tau2) {$p_2$};
				
				\draw [line width = 0.3mm, dotted] (root0) -- (tau);
				
				\draw [line width = 0.3mm, dotted] (root0) -- (tau0);
				\draw [line width = 0.3mm, dotted] (root1) -- (tau1);
				\draw [line width = 0.3mm, dotted] (root2) -- (tau2);
			\end{scope}
			
			\begin{scope}[canvas is yz plane at x=0]
				\draw [line width = 0.4mm, color = gray!80]  (0,0) node (root00) {} -- (0,2.2);
				\draw [line width = 0.4mm, color = gray!80]  (3,0) node (root01) {} -- (3,2.2);
				\draw [line width = 0.4mm, color = gray!80]  (6,0) node (root02) {} -- (6,2.2);
				
				\draw [draw = none, fill = gray!80] (root00) [above=1cm] circle [radius=1mm] node (dot001) {};
				\draw [draw = none, fill = gray!80] (root00) [above=2cm] circle [radius=1mm] node (dot002) {};
				
				\draw [draw = none, fill = gray!80] (root01) [above=1cm] circle [radius=1mm] node (dot011) {};
				\draw [draw = none, fill = gray!80] (root01) [above=2cm] circle [radius=1mm] node (dot012) {};
				
				\draw [draw = none, fill = gray!80] (root02) [above=1cm] circle [radius=1mm] node (dot021) {};
				\draw [draw = none, fill = gray!80] (root02) [above=2cm] circle [radius=1mm] node (dot022) {};
				
				\draw [draw = none] (0,2.8) node (tau00) {$p_{0,0}$};
				\draw [draw = none] (3,2.8) node (tau01) {$p_{0,1}$};
				\draw [draw = none] (6,2.8) node (tau02) {$p_{0,2}$};
				
				\draw [line width = 0.3mm, dotted, color = gray!80] (root00) -- (tau00);
				\draw [line width = 0.3mm, dotted, color = gray!80] (root01) -- (tau01);
				\draw [line width = 0.3mm, dotted, color = gray!80] (root02) -- (tau02);
				
				\draw [draw = none] ($(root00) + (-0.75,0.4)$) node (run00) {$r_{0,0}[0]$};
				\draw [draw = none] ($(dot001) + (-0.6,0.1)$) node (run01) {$r_{0,0}[1]$};
				\draw [draw = none] ($(dot002) + (-0.6,0.1)$) node (run01) {$r_{0,0}[2]$};
				
				\draw [draw = none] ($(root01) + (-0.75,0.4)$) node (run01) {$r_{0,1}[0]$};
				\draw [draw = none] ($(dot011) + (-0.6,0.1)$) node (run01) {$r_{0,1}[1]$};
				\draw [draw = none] ($(dot012) + (-0.6,0.1)$) node (run01) {$r_{0,1}[2]$};
			\end{scope}
			
			\begin{scope}[canvas is yz plane at x=1.1]
 				\draw [line width = 0.65mm, color = white]  (0,0) -- (0,2.2);
				\draw [line width = 0.65mm, color = white]  (3,0) -- (3,2.2);
				\draw [line width = 0.65mm, color = white]  (6,0) -- (6,2.2);
				\draw [line width = 0.4mm, color = gray!80]  (0,0) node (root10) {} -- (0,2.2);
				\draw [line width = 0.4mm, color = gray!80]  (3,0) node (root11) {} -- (3,2.2);
				\draw [line width = 0.4mm, color = gray!80]  (6,0) node (root12) {} -- (6,2.2);
				
				\draw [draw = none, fill = gray!80] (root10) [above=1cm] circle [radius=1mm] node (dot101) {};
				\draw [draw = none, fill = gray!80] (root10) [above=2cm] circle [radius=1mm] node (dot102) {};
				
				\draw [draw = none, fill = gray!80] (root11) [above=1cm] circle [radius=1mm] node (dot111) {};
				\draw [draw = none, fill = gray!80] (root11) [above=2cm] circle [radius=1mm] node (dot112) {};
				
				\draw [draw = none, fill = gray!80] (root12) [above=1cm] circle [radius=1mm] node (dot121) {};
				\draw [draw = none, fill = gray!80] (root12) [above=2cm] circle [radius=1mm] node (dot122) {};
				
				\draw [draw = none] (0,2.8) node (tau10) {$p_{1,0}$};
				\draw [draw = none] (3,2.8) node (tau11) {$p_{1,1}$};
				\draw [draw = none] (6,2.8) node (tau12) {$p_{1,2}$};
				
				\draw [line width = 0.3mm, dotted, color = gray!80] (root10) -- (tau10);
				\draw [line width = 0.3mm, dotted, color = gray!80] (root11) -- (tau11);
				\draw [line width = 0.3mm, dotted, color = gray!80] (root12) -- (tau12);
				
				\draw [draw = none] ($(root10) + (-0.75,0.4)$) node (run10) {$r_{1,0}[0]$};
			\end{scope}
			
			\begin{scope}[canvas is yz plane at x=2.2]
 				\draw [line width = 0.65mm, color = white]  (0,0) -- (0,2.2);
				\draw [line width = 0.65mm, color = white]  (3,0) -- (3,2.2);
				\draw [line width = 0.65mm, color = white]  (6,0) -- (6,2.2);
				\draw [line width = 0.4mm, color = gray!80]  (0,0) node (root20) {} -- (0,2.2);
				\draw [line width = 0.4mm, color = gray!80]  (3,0) node (root21) {} -- (3,2.2);
				\draw [line width = 0.4mm, color = gray!80]  (6,0) node (root22) {} -- (6,2.2);
				
				\draw [draw = none, fill = gray!80] (root20) [above=1cm] circle [radius=1mm] node (dot201) {};
				\draw [draw = none, fill = gray!80] (root20) [above=2cm] circle [radius=1mm] node (dot202) {};
				
				\draw [draw = none, fill = gray!80] (root21) [above=1cm] circle [radius=1mm] node (dot211) {};
				\draw [draw = none, fill = gray!80] (root21) [above=2cm] circle [radius=1mm] node (dot212) {};
				
				\draw [draw = none, fill = gray!80] (root22) [above=1cm] circle [radius=1mm] node (dot221) {};
				\draw [draw = none, fill = gray!80] (root22) [above=2cm] circle [radius=1mm] node (dot222) {};
				
				\draw [draw = none] (0,2.8) node (tau20) {$p_{2,0}$};
				\draw [draw = none] (3,2.8) node (tau21) {$p_{2,1}$};
				\draw [draw = none] (6,2.8) node (tau22) {$p_{2,2}$};
				
				\draw [line width = 0.3mm, dotted, color = gray!80] (root20) -- (tau20);
				\draw [line width = 0.3mm, dotted, color = gray!80] (root21) -- (tau21);
				\draw [line width = 0.3mm, dotted, color = gray!80] (root22) -- (tau22);
			\end{scope}
			
			\begin{scope}[canvas is xy plane at z=0]
				\draw [draw = none, fill = black] (root0) circle [radius=0.7mm] node (dot000) {};
				\draw [draw = none, fill = black] (root0) [right=3cm] circle [radius=0.7mm] node (dot010) {};
				\draw [draw = none, fill = black] (root0) [right=6cm] circle [radius=0.7mm] node (dot020) {};
				
				\draw [draw = none, fill = black] (root1) circle [radius=0.7mm] node (dot100) {};
				\draw [draw = none, fill = black] (root1) [right=3cm] circle [radius=0.7mm] node (dot110) {};
				\draw [draw = none, fill = black] (root1) [right=6cm] circle [radius=0.7mm] node (dot120) {};
				
				\draw [draw = none, fill = black] (root2) circle [radius=0.7mm] node (dot200) {};
				\draw [draw = none, fill = black] (root2) [right=3cm] circle [radius=0.7mm] node (dot210) {};
				\draw [draw = none, fill = black] (root2) [right=6cm] circle [radius=0.7mm] node (dot220) {};
			\end{scope}
		\end{tikzpicture}
	}}
	\caption{A $3$-dimensional comb graph resulting from the arrangement of witnesses for a HyperCTL$^*$ formula $\exists \pi. \square (\exists \pi'. \square (\exists \pi''. \varphi))$. 
	The innermost witnesses $p_{i,j}$ are labeled with the automaton states of the corresponding automaton runs.}
\label{3dcomb}
\end{figure}
	
	In the $d$-dimensional case, diagonals are hyperplanes.
	We represent the $k$-th plane $D_k$ in the $d$-dimensional comb by a nested sequence of depth $d-1$.
	\begin{align*}
	D_k \coloneqq & [D_{k, 0}, \ldots , D_{k,k}]\\
	D_{k, i_1} \coloneqq & [D_{k,i_1, 0} \ldots , D_{k, i_1, k-i_1}]\\[-0.5em]
	& \qquad \vdots \\[-0.5em]
	D_{k, i_1, \ldots, i_{d-2}} \coloneqq & [ p_{i_1, \ldots, i_{d-2}, 0}[s], p_{i_1, \ldots, i_{d-2}, 1}[s-1], \ldots\\
	& p_{i_1, \ldots, i_{d-2}, s}[0] ] 
	 \text{\scriptsize , where $s = k - (i_1 + \ldots + i_{d-2})$ }
	\end{align*} 
	We additionally define $d$-dimensional frontiers as nested sets of automata states.
	\begin{alignat*}{2}
		F_k &\coloneqq \set{F_{k, i_1} ~|~ i_1 \leq k, i_1 \in \nat} \\	
		F_{k, i_1} &\coloneqq \set{F_{k, i_1, i_2} ~|~ i_1 + i_2 \leq k, i_2 \in \nat} \\[-0.5em]
		& \qquad \qquad \vdots \\[-0.5em]
		F_{k, i_1, \ldots, i_{d-2}} &\coloneqq \{ r_{i_1, \ldots, i_{d-2}, i_{d-1}}[i_d] ~|~ \\
		& \qquad i_{d_1} + i_d = k - (i_1 + \ldots + i_{d-2}) \} 
	\end{alignat*}
	As an example, in Fig.~\ref{3dcomb}, $[p_{0,0}[1], p_{0,1}[0]]$, and $[p_{1,0}[0]]$ constitute plane $D_1$.
	The corresponding frontier is giving by $F_1 = \set{F_{1,0}, F_{1,1}} = \set{\set{r_{0,0}[1], r_{0,1}[0]}, \set{r_{1,0}[0]}}$.

	We define the \emph{cuttable} property recursively, with the definition for the $2$-dimensional case (c.f. Lemma~\ref{lem:HyperCTL-SAT-Release}) as base case.
	For indices $i_1, \ldots , i_{l}$ we also write $\bar i$.
	Two sub-planes $D_{k,\bar {i}}$, $D_{k',\bar {i'}}$ with $F_{k,\bar{i}} = F_{k',\bar{i'}}$ are \emph{cuttable} if for every two $F_{k', \bar {i'}, j'} \in F_{k', \bar {i'}}$ and $F_{k, \bar {i}, j} \in F_{k, \bar {i}}$ with $F_{k', \bar {i'}, j'} = F_{k, \bar {i}, j}$, sub-frontier $F_{k, \bar {i}, j} $ is associated with at least as many sub-sequences in $D_{k,\bar {i}}$ as $F_{k', \bar {i'}}$ is associated with in $D_{k',\bar {i'}}$, or at least with $|Q|$ many (where $Q$ is the set of automata states).
	Furthermore, if $F_{k, \bar {i}, j}$ is associated with $D_{k, \bar {i}, j}$, and $F_{k', \bar {i'}, j'}$ is associated with $D_{k', \bar {i'}, j'} $, the corresponding sub-planes must be cuttable again.
	
	For the sake of explainability and compactness, we extend the \emph{cut} operation only to the $3$-dimensional case.
	The $d$-dimensional case generalizes the definition.
	A plane in the $3$-dimensional comb is a sequence of sequences $D_{k, i}$ of nodes. 
	Each sequence $D_{k, i}$ contains nodes that reside on paths branching from $p_{i}$.
	To cut plane $D_{k'}$ to $D_k$, we require that $F_k = F_{k'}$.
	Pick for each set $F_{k, i}$ an equal set $F_{k', i'}$ and cut the diagonal $D_{k',i'}$ to $D_{k,i}$, as described for the $2$-dimensional case in the proof of Lemma~\ref{lem:HyperCTL-SAT-Release}.
	To preserve the relations of paths in the third dimension ($p_{i,j}$) with the paths in the first and second dimension ($p$ and $p_i$), also replace each $2$-dimensional sub-comb with origin at $p_{i}[k-i]$ with the sub-comb at $p_{i'}[k'-i']$; and the $3$-dimensional sub-comb with origin at $p[k]$ with the sub-comb $p[k']$.
	Using the lifted definition of cuttable, it is also possible to define \emph{preserving cuts} by first declaring a set of representative positions and then choosing the sets in such a way that no representative positions are deleted during the cut.

	With the same arguments as in the $2$-dimensional case and using the lifted (preserving) cut operation, we can create a $3$-dimensional comb prefix $\mathcal{T_{\mathit{fin}}}$ of bounded size which can then be extended to a satisfying model $\mathcal{\tilde T}$. 
	Note that the bound in the $3$-dimensional case is exponentially larger than the bound in the $2$-dimensional case due to the more complicated definition of cuttable.
	\end{proof}

\begin{thm}\label{HyperCTLStarSAT}
	The satisfiability problem for HyperCTL$^*$ formulas in the $\exists^*\forall^*$~fragment is decidable.
\end{thm}

\begin{proof}
	The proof is similar to the proof for Theorem~\ref{hyperqptl-sat-decidable}. 
	Let an $\exists^* \forall^*$ HyperCTL$^*$ formula $\varphi$ be given with at least one existential quantifier, the case with no existential quantifier is handled by Lemma~\ref{lem:HyperCTL-SAT-forall}. 
	We again eliminate the universal quantification by explicitly enumerating every possible interaction between the universal and existential quantifiers. 
	The resulting formula is in the $\exists^*$ fragment, which can be decided (cf. Lemma~\ref{hyperCTL$^*$-exists-decidabel}).
\end{proof}

\begin{thm}\label{HyperCTLStarUndec}
	The satisfiability problem for HyperCTL$^*$ formulas in the $\forall \exists$~fragment is undecidable.
\end{thm}

\begin{proof}
	HyperCTL$^*$ subsumes HyperLTL as a syntactic fragment. 
	As in Theorem~\ref{thm:HyperQPTLUndec}, the undecidability of the $\forall\exists$~HyperCTL$^*$ fragment follows from the undecidability of the $\forall\exists$~HyperLTL fragment~\cite{DecidingHyperLTL}.
\end{proof}
We have thus identified the largest decidable fragment of HyperCTL$^*$, which is, as for HyperLTL and HyperQPTL, the $\exists^* \forall^*$~fragment. 

\section{Conclusion and Future Work}
\label{conclusion}

In this paper, we have developed a comprehensive expressiveness hierarchy of linear-time and branching-time hyperlogics.
In the linear-time spectrum, we studied the relationships between HyperLTL, \foe, HyperQPTL and \seins.
In the branching-time spectrum, we studied HyperCTL$^*$, \mple, HyperQCTL$^*$ and \msoe.
Our results show significant differences between the hierarchy of the standard logics and the hierarchy of the hyperlogics. 
Most significantly, the equivalences between QPTL and S1S, and CTL$^*$ and MPL do not translate to the corresponding hyperlogics: HyperQPTL is not equivalent to \seins, and HyperCTL$^*$ is not equivalent to \mple. 
The only equivalence that translates to the hierarchy of hyperlogics is the equivalence between QCTL$^*$ and MSO: HyperQCTL$^*$ and \msoe are equivalent. 
The reason for this equivalence is that on \emph{trees}, it is possible to simulate the equal-level predicate and the second-order quantification by using propositional HyperQCTL$^*$ quantification (as described in the proof of Theorem~\ref{thm:HyperQCTL*equivMSOE}).

Our results also identify the decidability boundary of the satisfiability problem.
The previous result here was that HyperLTL formulas in the $\exists^* \forall^*$~fragment are decidable~\cite{DecidingHyperLTL}. 
We showed that the same holds for the more expressive logic HyperQPTL and that, hence, propositional quantification does not influence the decidability of the satisfiability problem in the linear-time hierarchy. 
In the branching-time hierarchy, we identified the largest decidable fragment of HyperCTL$^*$, which is, again, the $\exists^* \forall^*$ fragment.

Our results indicate that the  equal-level predicate adds, in most cases,  more expressiveness than trace or path variables. 
This raises the question if alternatives to the equal-level predicate can be found that would lift logics like \fo to hyperlogics while preserving the equivalences to the temporal logics. 
Another direction for future work concerns
the extension of the temporal logics with knowledge modalities~\cite{FaginHMV95}. It is known that
KLTL and KCTL$^*$, the extensions of LTL and CTL$^*$ with knowledge modalities under perfect recall semantics, are not subsumed by 
HyperLTL and HyperCTL$^*$, respectively~\cite{Bozzelli}. 
The question how extensions of HyperLTL and HyperCTL$^*$ with knowledge modalities fit into the expressiveness hierarchy is open.

\bibliographystyle{IEEEtran}
\bibliography{main.bib}

\begin{thebibliography}{10}
\providecommand{\url}[1]{#1}
\csname url@samestyle\endcsname
\providecommand{\newblock}{\relax}
\providecommand{\bibinfo}[2]{#2}
\providecommand{\BIBentrySTDinterwordspacing}{\spaceskip=0pt\relax}
\providecommand{\BIBentryALTinterwordstretchfactor}{4}
\providecommand{\BIBentryALTinterwordspacing}{\spaceskip=\fontdimen2\font plus
\BIBentryALTinterwordstretchfactor\fontdimen3\font minus
  \fontdimen4\font\relax}
\providecommand{\BIBforeignlanguage}[2]{{%
\expandafter\ifx\csname l@#1\endcsname\relax
\typeout{** WARNING: IEEEtran.bst: No hyphenation pattern has been}%
\typeout{** loaded for the language `#1'. Using the pattern for}%
\typeout{** the default language instead.}%
\else
\language=\csname l@#1\endcsname
\fi
#2}}
\providecommand{\BIBdecl}{\relax}
\BIBdecl

\bibitem{LTL}
\BIBentryALTinterwordspacing
A.~Pnueli, ``The temporal logic of programs,'' in \emph{18th Annual Symposium
  on Foundations of Computer Science, Providence, Rhode Island, USA, 31 October
  - 1 November 1977}, 1977, pp. 46--57. [Online]. Available:
  \url{https://doi.org/10.1109/SFCS.1977.32}
\BIBentrySTDinterwordspacing

\bibitem{CTLStar}
\BIBentryALTinterwordspacing
E.~A. Emerson and J.~Y. Halpern, ``"sometimes" and "not never" revisited: on
  branching versus linear time temporal logic,'' \emph{J. {ACM}}, vol.~33,
  no.~1, pp. 151--178, 1986. [Online]. Available:
  \url{https://doi.org/10.1145/4904.4999}
\BIBentrySTDinterwordspacing

\bibitem{SpectrumOfTemporalLogics}
B.~Finkbeiner and M.~N. Rabe, ``The linear-hyper-branching spectrum of temporal
  logics,'' \emph{it - Information Technology}, vol.~56, pp. 273--279, November
  2014.

\bibitem{HyperLTL}
\BIBentryALTinterwordspacing
M.~R. Clarkson, B.~Finkbeiner, M.~Koleini, K.~K. Micinski, M.~N. Rabe, and
  C.~S{\'{a}}nchez, ``Temporal logics for hyperproperties,'' in
  \emph{Principles of Security and Trust - Third International Conference,
  {POST} 2014, Held as Part of the European Joint Conferences on Theory and
  Practice of Software, {ETAPS} 2014, Grenoble, France, April 5-13, 2014,
  Proceedings}, 2014, pp. 265--284. [Online]. Available:
  \url{https://doi.org/10.1007/978-3-642-54792-8\_15}
\BIBentrySTDinterwordspacing

\bibitem{Goguen+Meseguer/1982/SecurityPoliciesAndSecurityModels}
\BIBentryALTinterwordspacing
J.~A. Goguen and J.~Meseguer, ``Security policies and security models,'' in
  \emph{1982 {IEEE} Symposium on Security and Privacy, Oakland, CA, USA, April
  26-28, 1982}, 1982, pp. 11--20. [Online]. Available:
  \url{https://doi.org/10.1109/SP.1982.10014}
\BIBentrySTDinterwordspacing

\bibitem{Zdancewic+Myers/03/ObservationalDeterminism}
\BIBentryALTinterwordspacing
S.~Zdancewic and A.~C. Myers, ``Observational determinism for concurrent
  program security,'' in \emph{16th {IEEE} Computer Security Foundations
  Workshop {(CSFW-16} 2003), 30 June - 2 July 2003, Pacific Grove, CA, {USA}},
  2003, p.~29. [Online]. Available:
  \url{https://doi.org/10.1109/CSFW.2003.1212703}
\BIBentrySTDinterwordspacing

\bibitem{Thomas09}
W.~Thomas, ``Path logics with synchronization,'' in \emph{Perspectives in
  Concurrency Theory}, K.~Lodaya, M.~Mukund, and R.~Ramanujam, Eds.\hskip 1em
  plus 0.5em minus 0.4em\relax IARCS-Universities, Universities Press, 2009,
  pp. 469--481.

\bibitem{Martin}
\BIBentryALTinterwordspacing
B.~Finkbeiner and M.~Zimmermann, ``The first-order logic of hyperproperties,''
  in \emph{34th Symposium on Theoretical Aspects of Computer Science, {STACS}
  2017, March 8-11, 2017, Hannover, Germany}, 2017, pp. 30:1--30:14. [Online].
  Available: \url{https://doi.org/10.4230/LIPIcs.STACS.2017.30}
\BIBentrySTDinterwordspacing

\bibitem{Kamp68}
H.~W. Kamp, ``Tense logic and the theory of linear order,'' Ph.D. dissertation,
  Computer Science Department, University of California at Los~Angeles, USA,
  1968.

\bibitem{KampGabbay}
\BIBentryALTinterwordspacing
D.~Gabbay, A.~Pnueli, S.~Shelah, and J.~Stavi, ``On the temporal analysis of
  fairness,'' in \emph{Proceedings of the 7th ACM SIGPLAN-SIGACT Symposium on
  Principles of Programming Languages}, ser. POPL '80.\hskip 1em plus 0.5em
  minus 0.4em\relax New York, NY, USA: ACM, 1980, pp. 163--173. [Online].
  Available: \url{http://doi.acm.org/10.1145/567446.567462}
\BIBentrySTDinterwordspacing

\bibitem{QPTL}
A.~P. Sistla, ``Theoretical issues in the design and verification of
  distributed systems,'' Ph.D. dissertation, 1983.

\bibitem{QPTL-S1S}
\BIBentryALTinterwordspacing
Y.~Kesten and A.~Pnueli, ``A complete proof systems for {QPTL},'' in
  \emph{Proceedings, 10th Annual {IEEE} Symposium on Logic in Computer Science,
  San Diego, California, USA, June 26-29, 1995}, 1995, pp. 2--12. [Online].
  Available: \url{https://doi.org/10.1109/LICS.1995.523239}
\BIBentrySTDinterwordspacing

\bibitem{moller1999}
\BIBentryALTinterwordspacing
F.~Moller and A.~M. Rabinovich, ``On the expressive power of {CTL},'' in
  \emph{14th Annual {IEEE} Symposium on Logic in Computer Science, Trento,
  Italy, July 2-5, 1999}.\hskip 1em plus 0.5em minus 0.4em\relax {IEEE}
  Computer Society, 1999, pp. 360--368. [Online]. Available:
  \url{https://doi.org/10.1109/LICS.1999.782631}
\BIBentrySTDinterwordspacing

\bibitem{abrahamson1980}
K.~R. Abrahamson, ``Decidability and expressiveness of logics of processes,''
  Ph.D. dissertation, Seattle, WA, USA, 1980, aAI8109709.

\bibitem{QCTLStar}
\BIBentryALTinterwordspacing
T.~French, ``Decidability of quantifed propositional branching time logics,''
  in \emph{{AI} 2001: Advances in Artificial Intelligence, 14th Australian
  Joint Conference on Artificial Intelligence, Adelaide, Australia, December
  10-14, 2001, Proceedings}, 2001, pp. 165--176. [Online]. Available:
  \url{https://doi.org/10.1007/3-540-45656-2\_15}
\BIBentrySTDinterwordspacing

\bibitem{laroussinie2014}
\BIBentryALTinterwordspacing
F.~Laroussinie and N.~Markey, ``Quantified {CTL:} expressiveness and
  complexity,'' \emph{Logical Methods in Computer Science}, vol.~10, no.~4,
  2014. [Online]. Available: \url{https://doi.org/10.2168/LMCS-10(4:17)2014}
\BIBentrySTDinterwordspacing

\bibitem{Kupferman2009}
\BIBentryALTinterwordspacing
O.~Kupferman, N.~Piterman, and M.~Y. Vardi, ``From liveness to promptness,''
  \emph{Formal Methods in System Design}, vol.~34, no.~2, pp. 83--103, Apr
  2009. [Online]. Available: \url{https://doi.org/10.1007/s10703-009-0067-z}
\BIBentrySTDinterwordspacing

\bibitem{LTL-SAT}
\BIBentryALTinterwordspacing
A.~P. Sistla and E.~M. Clarke, ``The complexity of propositional linear
  temporal logics,'' \emph{J. {ACM}}, vol.~32, no.~3, pp. 733--749, 1985.
  [Online]. Available: \url{https://doi.org/10.1145/3828.3837}
\BIBentrySTDinterwordspacing

\bibitem{DecidingHyperLTL}
\BIBentryALTinterwordspacing
B.~Finkbeiner and C.~Hahn, ``Deciding hyperproperties,'' in \emph{27th
  International Conference on Concurrency Theory, {CONCUR} 2016, August 23-26,
  2016, Qu{\'{e}}bec City, Canada}, 2016, pp. 13:1--13:14. [Online]. Available:
  \url{https://doi.org/10.4230/LIPIcs.CONCUR.2016.13}
\BIBentrySTDinterwordspacing

\bibitem{CTLStar-SAT}
\BIBentryALTinterwordspacing
E.~A. Emerson and A.~P. Sistla, ``Deciding full branching time logic,''
  \emph{Information and Control}, vol.~61, no.~3, pp. 175--201, 1984. [Online].
  Available: \url{https://doi.org/10.1016/S0019-9958(84)80047-9}
\BIBentrySTDinterwordspacing

\bibitem{MarkusThesis}
\BIBentryALTinterwordspacing
M.~N. Rabe, ``A temporal logic approach to information-flow control,'' Ph.D.
  dissertation, Saarland University, 2016. [Online]. Available:
  \url{http://scidok.sulb.uni-saarland.de/volltexte/2016/6387/}
\BIBentrySTDinterwordspacing

\bibitem{sistla1987complementation}
\BIBentryALTinterwordspacing
A.~P. Sistla, M.~Y. Vardi, and P.~Wolper, ``The complementation problem for
  b{\"{u}}chi automata with applications to temporal logic (extended
  abstract),'' in \emph{Automata, Languages and Programming, 12th Colloquium,
  Nafplion, Greece, July 15-19, 1985, Proceedings}, 1985, pp. 465--474.
  [Online]. Available: \url{https://doi.org/10.1007/BFb0015772}
\BIBentrySTDinterwordspacing

\bibitem{mcnaughton1971}
R.~McNaughton and S.~A. Papert, \emph{{Counter-Free Automata (MIT research
  monograph no. 65)}}.\hskip 1em plus 0.5em minus 0.4em\relax The MIT Press,
  1971.

\bibitem{PrinciplesOfModelChecking}
C.~Baier and J.~Katoen, \emph{Principles of model checking}.\hskip 1em plus
  0.5em minus 0.4em\relax {MIT} Press, 2008.

\bibitem{fagin2004reasoning}
R.~Fagin, J.~Y. Halpern, Y.~Moses, and M.~Vardi, \emph{{Reasoning About
  Knowledge}}.\hskip 1em plus 0.5em minus 0.4em\relax MIT press, 2004.

\bibitem{halpern1989}
\BIBentryALTinterwordspacing
J.~Y. Halpern and M.~Y. Vardi, ``The complexity of reasoning about knowledge
  and time. i. lower bounds,'' \emph{J. Comput. Syst. Sci.}, vol.~38, no.~1,
  pp. 195--237, 1989. [Online]. Available:
  \url{https://doi.org/10.1016/0022-0000(89)90039-1}
\BIBentrySTDinterwordspacing

\bibitem{Bozzelli}
\BIBentryALTinterwordspacing
L.~Bozzelli, B.~Maubert, and S.~Pinchinat, ``Unifying hyper and epistemic
  temporal logics,'' in \emph{Foundations of Software Science and Computation
  Structures - 18th International Conference, FoSSaCS 2015, Held as Part of the
  European Joint Conferences on Theory and Practice of Software, {ETAPS} 2015,
  London, UK, April 11-18, 2015. Proceedings}, 2015, pp. 167--182. [Online].
  Available: \url{https://doi.org/10.1007/978-3-662-46678-0\_11}
\BIBentrySTDinterwordspacing

\bibitem{LTLAutomata}
\BIBentryALTinterwordspacing
M.~Y. Vardi and P.~Wolper, ``Reasoning about infinite computations,''
  \emph{Inf. Comput.}, vol. 115, no.~1, pp. 1--37, 1994. [Online]. Available:
  \url{https://doi.org/10.1006/inco.1994.1092}
\BIBentrySTDinterwordspacing

\bibitem{FaginHMV95}
R.~Fagin, J.~Y. Halpern, Y.~Moses, and M.~Y. Vardi, \emph{Reasoning About
  Knowledge}.\hskip 1em plus 0.5em minus 0.4em\relax MIT Press, 1995.

\end{thebibliography}

\newpage

\appendices

\section{Proof of Lemma~\ref{MPLEsubsumesHyperKCTL*}}
\label{App:proofMPLEsubsumesHyperKCTL*}
\subsection{\mple Subsumes HyperCTL$^*$}
Like in Lemma~\ref{S1SESubsumesHyperQPTL}, we use the idea from~\cite{Martin} where $\set{y_1, y_2, \ldots}$ are first-order variables used to indicate time. We use second-order quantification to quantify traces.
Given a HyperKCTL$^*$ formula $\varphi$, we inductively construct an \mple formula. 
\begin{alignat*}{3}
&\text{mpe}(a_{\pi}, y_i) &&=~&& \exists x. ~x \in X_\pi \land E(x, y_i) \land P_a(x) \\
\noalign{\text{where $X_\pi$ is the second-order variable used for path $\pi$}}
&\text{mpe}(\neg \varphi_1, y_i) &&=&& \neg \text{mpe}(\varphi_1, y_i) \\
&\text{mpe}(\varphi_1 \lor \varphi_2, y_i) &&=&& \text{mpe}(\varphi_1, y_i) \lor \text{mpe}(\varphi_2, y_i) \\
&\text{mpe}(\X \varphi_1, y_i) &&=&& \exists y_j > y_i.~ \neg (\exists y_k.~  y_i < y_k < y_j) \\
& && && \quad \land \text{mpe}(\varphi_1, y_j) \\
&\text{mpe}(\varphi_1 \LTLuntil \varphi_2, y_i) &&=&& \exists y_j.~ y_i \leq y_j \land \text{mpe}(\varphi_2, y_j) \\
& && && \quad \land \forall y_k.~ y_i \leq y_k < y_j \\
& && && \quad \rightarrow \text{mpe}(\varphi_1, y_k) \\
&\text{mpe}(\exists \pi.\varphi_1, y_i) &&=&& \exists X_\pi. ~ \mathit{prefix}(X_\pi, X_\varepsilon, y_i) \\
& && && \quad \land \text{mpe}(\varphi_1, y_i) \\
\noalign{\text{where $X_\pi$ is now used as trace variable for $\pi$ } }
& \text{where} && && \\
& \mathit{prefix}(X_\pi, X_\varepsilon, y_i) &&\coloneqq&& ~ \forall y'_i \leq y_i.~ \forall x.~ E(x, y'_i) \rightarrow \\
& && && \quad (x \in X_\pi \leftrightarrow x \in X_\varepsilon) \\
& && && \quad \rightarrow x =_A  x'
\end{alignat*}
where we use $X_{\varepsilon}$ to denote the second-order variable that was quantified most recently (i.e. closest in the scope to $X_\pi$).
The \mple formula $\text{mpe}(\varphi, 0)$ is equivalent to the HyperCTL$^*$ formula $\varphi$, which can be shown by a straightforward induction.

\subsection{\msoe Can Express the Knowledge Modality}
Extend the translation from HyperCTL$^*$ to \mple given above with an additional rule for $\K_{A,\pi}$.
\begin{alignat*}{3}
	&\text{mpe}(\K_{A,\pi}.\varphi_1, y_i) &&=&& \forall X'_\pi.~ \mathit{eqPre}(X_\pi, X'_\pi, y_i, A) \\
	& && && \quad \rightarrow \text{mpe}(\varphi_1, y_i) \\
	\noalign{ \text{where $X'_\pi$ is now used as trace variable for $\pi$} }
	& \text{where} && && \\
	& \mathit{eqPre}(X_\pi, X'_\pi, y_i, A) &&\coloneqq&& ~ \forall x \in X_\pi. \forall x' \in X'_\pi.~ (\exists y'_i \leq y_i. \\
	& && && \quad ~ E(x, y'_i) \land E(x', y'_i)) \\
	& && && \quad \rightarrow x =_A  x'
\end{alignat*}

\section{Proof of Theorem~\ref{thm:HyperQCTL*equivMSOE}}
\label{App:proofHyperQCTL*equivMSOE}
\subsection{HyperQCTL$^*$ Subsumes \msoe}
Let $\varphi$ be a \msoe formula over $AP$. We define an equivalent HyperQCTL$^*$ formula hqc$(\varphi)$ as follows.
	\begin{alignat*}{3}
&\text{hqc}(P_a(x)) &&=~&& \LTLfinally (q^x_{\pi_x} \land a_{\pi_x}) \\
&\text{hqc}(x \in X) &&=~&& \LTLfinally (q^x_{\pi_x} \land q^X_{\pi_x}) \\
&\text{hqc}(x < y) &&=~&& \LTLfinally (q^x_{\pi_x} \land \LTLnext \LTLfinally q^y_{\pi_y}) \\
&\text{hqc}(x = y) &&=~&& \LTLfinally (q^x_{\pi_x} \land  q^y_{\pi_y}) \land \forall q. \LTLglobally (q_{\pi_{x}} \leftrightarrow q_{\pi_{y}}) \\
&\text{hqc}(E(x,y)) &&=~&& \LTLfinally (q^x_{\pi_x} \land  q^y_{\pi_y}) \\
&\text{hqc}(\neg \varphi_1) &&=~&& \neg \text{hqc}(\varphi_1) \\
&\text{hqc}(\varphi_1 \lor \varphi_2) &&=~&& \text{hqc}(\varphi_1) \lor \text{hqc}(\varphi_2) \\
&\text{hqc}(\exists x. \varphi_1) &&=~&& \exists \pi_x.~ \exists q^x.~ (\neg q^x_{\pi_x}) \LTLuntil (q^x_{\pi_x} \land \LTLnext\LTLglobally(\neg q^x_{\pi_x})) \\
& && && \phantom{ \exists \pi_x.~ \exists q^x. } \land \text{hqc}(\varphi_1) \\
&\text{hqc}(\exists X. \varphi_1) &&=~&& \exists q^X.~ \text{hqc}(\varphi_1)
\end{alignat*}
Note that we used the fact that two paths are equal in HyperQCTL$^*$ iff a universal quantification of atomic proposition $q$ always assigns $q$ to be globally equivalent on the two paths. Using a straightforward induction, we can show that for every tree $\mathcal{T}$, $\mathcal{T} \models \varphi$ iff $\mathcal{T} \models \text{hqc}(\varphi)$.

\subsection{\msoe Subsumes HyperQCTL$^*$}
Given a HyperQCTL$^*$ formula $\varphi$, we inductively construct an \msoe formula $\psi$. The construction is very similar to the one described in Appendix~\ref{App:proofMPLEsubsumesHyperKCTL*}. We only need to make sure that the second-order quantification $X_\pi$, which encodes a quantified path $\pi$, encodes a full path of the tree.
\begin{alignat*}{3}
&\text{mse}(a_{\pi}, y_i) &&=~&& \exists x.~ x \in X_\pi \land E(x, y_i) \land P_a(x) \\
 \noalign{ \text{where $X_\pi$ is the second-order variable used for path $\pi$} }
 \noalign{ \text{and $a$ is not quantified by a propositional quantifier.} }
&\text{mse}(q_{\pi}, y_i) &&=~&& \exists x.~ x \in X_\pi \land E(x, y_i) \land x \in X_q \\
\noalign{ \text{where $X_\pi$ is the second-order variable used for path $\pi$} }
\noalign{ \text{and $q$ is quantified by a propositional quantifier.} }
&\text{mse}(\neg \varphi_1, y_i) &&=&& \neg \text{mse}(\varphi_1, y_i) \\
&\text{mse}(\varphi_1 \lor \varphi_2, y_i) &&=&& \text{mse}(\varphi_1, y_i) \lor \text{mse}(\varphi_2, y_i) \\
&\text{mse}(\X \varphi_1, y_i) &&=&& \exists y_j.~ y_i < y_j \land \neg (\exists y_k.~ y_i < y_k < y_j) \\
& && && \land \text{mse}(\varphi_1, y_j) \\
&\text{mse}(\varphi_1 \LTLuntil \varphi_2, y_i) &&=&& \exists y_j.~ y_i \leq y_j \land \text{mse}(\varphi_2, y_j) \\
& && && \land \forall y_k.~ y_i \leq y_k < y_j \rightarrow \text{mse}(\varphi_1, y_k) \\
&\text{mse}(\exists \pi.\varphi_1, y_i) &&=&& \exists X_\pi.~ \mathit{path}(X_\pi) \land \\
& && && \mathit{prefix}(X_\pi, X_\varepsilon, y_i) \land \text{mse}(\varphi_1, y_i) \\
\noalign{ \text{where $X_\pi$ is now used as trace variable for $\pi$} }
&\text{mse}(\exists q.\varphi_1, y_i) &&=&& \exists X_q. \text{mse}(\varphi_1, y_i) \\
\noalign{ \text{where $X_q$ is now used as propositional variable for $q$} }
& \text{where} && && \\
& \mathit{path}(X_\pi) &&\coloneqq&& (\exists x \in X_\pi.~ \neg \exists x'. x' < x)~ \land \\
& && && ~ \forall x \in X_\pi. ~ \exists x' \in X_\pi.~ \mathit{succ}(x, x') ~ \land \\
& && && \quad \neg( \exists x'' \in X_\pi.~ \mathit{succ}(x, x'')) \\
& \mathit{succ}(x, x') &&\coloneqq&& x < x' \land \neg( \exists x''.~ x < x'' < x') \\
& \mathit{prefix}(X_\pi, X_\varepsilon, y_i) &&\coloneqq&& \forall y'_i \leq y_i.~ \forall x.~ E(x, y'_i) \rightarrow \\
& && && \quad (x \in X_\pi \leftrightarrow x \in X_\varepsilon)
\end{alignat*}
We use $X_{\varepsilon}$ to denote the second-order variable that was quantified most recently (i.e. closest in the scope to $X_\pi$) and also encodes a path.
Note that a set $X_\pi$ encodes a full path of the tree if it contains the root node and a unique direct successor for every node in the set.




%

\end{document}